\documentclass[a4paper,USenglish]{lipics}

\usepackage[T1]{fontenc}
\usepackage{bookmark}
\usepackage{xcolor}
\definecolor{linkcolor}{rgb}{0.9, 0.17, 0.31} %{0.57, 0.36, 0.51}
\definecolor{citecolor}{rgb}{0.36, 0.54, 0.66}

\hypersetup{unicode,bookmarksopen={true},bookmarksnumbered={true},colorlinks={true},citecolor=citecolor,linkcolor=linkcolor}
\usepackage{amssymb}
\usepackage{multicol}
\usepackage{bbm}
\newcommand{\nmathbbm}[1]{{\mbox{\usefont{U}{bbm}{m}{n}#1}}}
\usepackage{tikz}
\usetikzlibrary{decorations.text,arrows,snakes,shapes,fit,shadows,calc,patterns,intersections}
\tikzstyle{marked}=[draw,minimum size=1pt,circle,fill=white]
\tikzstyle{pebble}=[draw,circle,fill=black]
\tikzstyle{ligne}=[line width=0.5pt,-]

\usepackage{xspace}
\usepackage{enumitem}
\setlist[1]{itemsep=1.1ex plus .1 ex minus .1ex}
\newcommand\itemfmt[1]{\textcolor{darkgray}{\sffamily\bfseries\mathversion{bold}{#1}}}
\newcommand\itemref[1]{\textcolor{linkcolor}{\sffamily{#1}}}

\usepackage{microtype}

\newcommand{\efgame}{Ehrenfeucht-Fra\"iss\'e\xspace}
\newcommand\nat{\ensuremath{\mathbb{N}}\xspace}
\newcommand\rk[1]{\ensuremath{\textup{rank}(#1)}}

\newcommand\Abb{\ensuremath{\mathbb{A}}\xspace}
\newcommand\Lbb{\ensuremath{\mathbb{L}}\xspace}
\newcommand\Mbb{\ensuremath{\mathbb{M}}\xspace}
\newcommand\Kbb{\ensuremath{\mathbb{K}}\xspace}
\newcommand\Fbb{\ensuremath{\mathbb{F}}\xspace}

\newcommand\abb{\ensuremath{\nmathbbm{a}}\xspace}
\newcommand\wbb{\ensuremath{\nmathbbm{w}}\xspace}
\newcommand\ubb{\ensuremath{\nmathbbm{u}}\xspace}
\newcommand\vbb{\ensuremath{\nmathbbm{v}}\xspace}

\newcommand\Fs{\ensuremath{\mathcal{F}}\xspace}
\newcommand\Gs{\ensuremath{\mathcal{G}}\xspace}
\newcommand\Cs{\ensuremath{\mathcal{C}}\xspace}
\newcommand\Ss{\ensuremath{\mathcal{S}}\xspace}
\newcommand\Is{\ensuremath{\mathcal{I}}\xspace}

\newcommand\Vbf{\ensuremath{{\sf V}}\xspace}
\newcommand\Wbf{\ensuremath{{\sf W}}\xspace}
\newcommand\Dbf{\ensuremath{{\sf D}}\xspace}

\newcommand{\plus}{\ensuremath{+1,min,max}}

\newcommand{\sic}[1]{\ensuremath{\Sigma_{#1}}\xspace}
\newcommand{\sio}[1]{\ensuremath{\Sigma_{#1}(<)}\xspace}
\newcommand{\sip}[1]{\ensuremath{\Sigma_{#1}(<,\plus)}\xspace}

\newcommand{\pio}[1]{\ensuremath{\Pi_{#1}(<)}\xspace}

\newcommand{\bso}[1]{\ensuremath{\mathcal{B}\Sigma_{#1}(<)}\xspace}
\newcommand{\bsc}[1]{\ensuremath{\mathcal{B}\Sigma_{#1}}\xspace}
\newcommand{\bsp}[1]{\ensuremath{\mathcal{B}\Sigma_{#1}(<,\plus)}\xspace}

\newcommand{\sdp}{\ensuremath{\Sigma_{2}(<,+1)}\xspace}
\newcommand{\bdp}{\ensuremath{\mathcal{B}\Sigma_{2}(<,+1)}\xspace}
\newcommand{\stp}{\ensuremath{\Sigma_{3}(<,+1)}\xspace}
\newcommand{\pdp}{\ensuremath{\Pi_{2}(<,+1)}\xspace}
\newcommand{\ptp}{\ensuremath{\Pi_{3}(<,+1)}\xspace}
\newcommand{\sdo}{\sio{2}}

\newcommand{\fow}{\ensuremath{\textup{FO}(<)}\xspace}
\newcommand{\foeq}{\ensuremath{\textup{FO}(=)}\xspace}
\newcommand{\fod}{\ensuremath{\textup{FO}^2(<)}\xspace}
\newcommand{\foeqp}{\ensuremath{\textup{FO}(=,+1)}\xspace}
\newcommand{\fodp}{\ensuremath{\textup{FO}^2(<,+1)}\xspace}

\newcommand\fodeq[1]{\ensuremath{\equiv_{#1}}\xspace}
\newcommand\kfodeq{\fodeq{k}}
\newcommand\fodeqp[1]{\ensuremath{\equiv^{+1}_{#1}}\xspace}
\newcommand\kfodeqp{\fodeqp{k}}

\newcommand\bceq[1]{\ensuremath{\cong_{#1}}\xspace}
\newcommand\kbceq{\bceq{k}}
\newcommand\bceqp[1]{\ensuremath{\cong^{+1}_{#1}}\xspace}
\newcommand\kbceqp{\bceqp{k}}

\newcommand{\ucroch}[1]{\ensuremath{\left\lceil #1 \right\rceil}\xspace}
\newcommand{\croch}[1]{\ensuremath{\left\lfloor #1 \right\rfloor}\xspace}

\newcommand\wfA{\ensuremath{\Abb_\alpha}\xspace}

\newcommand\sieq[1]{\ensuremath{\preccurlyeq_{#1}}\xspace}
\newcommand\ksieq{\sieq{k}}
\newcommand\sieqp[1]{\ensuremath{\preccurlyeq^{+1}_{#1}}\xspace}
\newcommand\ksieqp{\sieqp{k}}

\newcommand\highlight[1]{\par\bigskip\noindent\textbf{\sffamily #1}.}

\theoremstyle{plain}
\newtheorem{fact}[theorem]{Fact}
\newtheorem{proposition}[theorem]{Proposition}
\newtheorem{clm}[theorem]{Claim}

\title{A Transfer Theorem for the Separation Problem\footnote{Supported by ANR 2010 BLAN 0202 01 FREC}}

\author{Thomas~Place}
\author{Marc~Zeitoun}

\affil{LaBRI, Bordeaux University, France, \texttt{firstname.lastname@labri.fr}.}

\authorrunning{T.~Place and M.~Zeitoun}

\Copyright{Thomas~Place and Marc~Zeitoun}

\subjclass{F.4.3 Formal Languages}

\keywords{Separation Problem, Regular Word Languages, Logics, Decidable Characterizations, Semidirect Product}

\begin{document}

\maketitle

\begin{abstract}
  We investigate two problems for a class \Cs of regular word languages. The
  \Cs-membership problem asks for an algorithm to decide whether an input
  language belongs to \Cs. The \Cs-separation problem asks for an algorithm
  that, given as input two regular languages, decides whether there exists a
  third language in \Cs containing the first language, while being disjoint
  from the second. These problems are considered as means to obtain a deep
  understanding of the class \Cs.

  It is usual for such classes to be defined by logical formalisms. Logics are
  often built on top of each other, by adding new predicates. A natural
  construction is to enrich a logic with the successor
  relation.     In this paper, we obtain new and simple proofs of two transfer
  results: we show that for suitable logically defined classes, the
  membership, resp.\ the separation problem for a class enriched with the
  successor relation reduces to the same problem for the original~class.

  Our reductions work both for languages of finite words and infinite
  words.   The proofs are mostly self-contained, and only require a basic background
  on regular languages. This paper therefore gives simple proofs of
  results that were considered as difficult, such as the decidability of the
  membership problem for the levels 1, 3/2, 2 and 5/2 of the dot-depth
  hierarchy.
\end{abstract}

\section{Introduction}
\label{sec:intro}
A central problem in the theory of formal languages is to
characterize and understand the expressive power of high level specification
formalisms. Monadic second order logic (MSO) is such a formalism, which is
both expressive and robust. For several classes of structures, such as words
or trees, it has the same expressive power as finite automata and defines the
class of regular languages. In this paper, we investigate fragments of MSO
over words. In this context, understanding the expressive power of a fragment
is associated to two decision problems: the \emph{membership problem} and the
\emph{separation problem}.

For a fixed logical fragment \Fs, the \emph{\Fs-membership
  problem} asks for a decision procedure that tests whether some input
regular language can be expressed by a formula from~\Fs. To obtain
such an algorithm, one has to consider and understand \emph{all}
properties that can be expressed within \Fs, which requires a deep
understanding of the fragment \Fs. On the other hand, the
\emph{\Fs-separation problem} is more general. It asks for
a decision procedure that tests whether given \emph{two} input
regular languages, there exists a third one in \Fs containing the
first language while being disjoint from the second one.

Since regular
languages are closed under complement, membership reduces to
separation: a language is in \Fs if and only if it can be separated from
its complement. Usually, the separation problem is more difficult than
the membership problem but also more rewarding with respect to the
knowledge gained on the investigated fragment \Fs.

These two problems have been considered and solved for many natural fragments
of monadic second order logic. Among these, the most prominent one is
first-order logic, \fow, equipped with a predicate $<$ for the linear
ordering. The solution to the membership problem, known as the
McNaughton-Papert-Schützenberger Theorem~\cite{sfo,mnpfo}, has been revisited
until recently~\cite{Diekert&Gastin:First-order-definable-languages:2008:a}.
The theorem states that a regular language is definable in \fow if and only if
its \emph{syntactic semigroup} is aperiodic. The syntactic semigroup is a
finite algebraic object that can be computed from any regular language. Since
aperiodicity can be defined as an equation that needs to be satisfied by all
of its elements, this yields decidability of \fow-definability. This result
now serves as a template, which is commonly followed in this line of research.

The separation problem has also been successfully solved for first-order logic~\cite{Henckell:Pointlike-sets:-finest-aperiodic:1988:a}.
Actually, the problem was first addressed in a purely algebraic
framework, and was later identified as equivalent to our separation
problem~\cite{MR1709911}. As for membership, this problem is still
revisited today and a new self-contained and combinatorial proof was
obtained in~\cite{PZ:lics14}.

\highlight{Motivation}  We are interested in natural fragments
of \fow obtained by restricting either the number of variables or the number
of quantifier alternations allowed in formulas. Such restrictions in general
give rise to several variants of the same fragment. Indeed, in most cases, the
drop in expressive power forbids the use of natural relations that could be
defined from the linear order in \fow. The main example considered in this
paper is $+1$, the \emph{successor relation}, together with predicates $min$
and $max$ for the first and last positions in a word. This means that one can
define two distinct variants of the same fragment depending on whether we
decide to explicitly add these predicates in the signature or
not. An example is the fragment $\sic n$, which consists of first-order formulas
whose prenex normal form has at most~$(n-1)$ quantifier alternations and starts with an
existential block. Since defining~$+1$ requires an additional
quantifier alternation, $\sip n$ has indeed stronger expressiveness than $\sio n$.
The motivation of this paper is to obtain decidability results for such
enriched fragments.

\highlight{State of the Art} Even when the weak fragment is
known to have decidable membership, proving that the enriched one has
the same property can be nontrivial. Examples include the membership
proofs of $\bsp1$ (Boolean combinations of $\sip1$ formulas) and
$\sdp$, which require difficult and intricate combinatorial
arguments~\cite{Knast:dd1:1983a,glasser-dd3/2,DBLP:journals/ijfcs/KufleitnerL12}
or a wealth of algebraic machinery~\cite{pwdelta,pw:wreath}. Another issue is that most proofs
directly deal with the enriched fragment. Given the jungle of such logical
fragments, it is desirable to avoid such an approach, treating each variant of
the same fragment independently. Instead, a satisfying approach is to first
obtain a solution of the membership and separation problems for the less
expressive variant and then to lift it to other variants via a generic
transfer result.

This approach has first been investigated by Straubing for the membership
problem~\cite{Str85} in an algebraic framework, and later adapted to be able
to treat classes not closed under complement~\cite{pw:wreath}. Transferring the logical
problem to this algebraic framework requires preliminary steps,
still \emph{specific} to the investigated class, to prove that:
\begin{enumerate}
  % \textbf{(1)}
\item\label{item:1} A language is definable in the fragment if and only if its syntactic semigroup
  belongs to a specific algebraic variety $\mathsf{V}$ (\emph{e.g.}, the
  variety of aperiodic monoids for~\fow), and
  % \textbf{(2)}
\item\label{item:2} Membership to~$\sf V$ is decidable.
\end{enumerate}
Next, though this is not immediate, for most fragments of \fow, it has been
proved that % \textbf{(3)}
\begin{enumerate}[resume]
\item\label{item:3} When the weaker variant corresponds to a variety {\sf V}, the variant
  with successor corresponds to the variety $\mathsf{V} \ast \mathsf{D}$,
  built generically from {\sf V}.
\end{enumerate}
Hence, Straubing's approach was to prove that
\begin{enumerate}[resume]
\item\label{item:4} % \textbf{(4)}
  the operator $\mathsf{V} \mapsto \mathsf{V}\ast\mathsf{D}$
  preserves decidability.
\end{enumerate}
Unfortunately, this is not true in
general~\cite{DBLP:journals/ijac/Auinger10}. Actually, while
decidability is preserved for all known logical fragments, there is no
generic result that captures them all. In particular, for the less
expressive fragments, one has to use completely \emph{ad hoc} proofs.
In the separation setting, things behave well: it
has been shown that decidability of separation is preserved by the operation $\mathsf{V}
\mapsto \mathsf{V} \ast \mathsf{D}$~\cite{Steinberg:delay-pointlikes:2001}.
While interesting when already starting from algebra, this approach has several downsides:
\begin{itemize}
\item Dealing with algebra hides the logical intuitions, while our primary
  goal is to understand the expressiveness of logics.
\item Going from logic to algebra requires to be acquainted with new notions
  and vocabulary, as well as involved theoretical tools. Proofs are also
  often nontrivial and require a deep understanding of complex objects, which
  may be scattered in the bibliography.
\item Despite step \ref{item:4}, % \textbf{(4)}
  which is generic to some extent, arguments
  specific to the investigated class are pushed to steps
  \ref{item:1}--\ref{item:3}, % \textbf{(1)--\textbf{(3)}}
  and they are often nontrivial.
\end{itemize}

\highlight{Contributions} We give a new
proof that decidability of separation can be transferred from a weak to an
enriched fragment. We present the result in two different forms.

The first one is non-algebraic: we work directly with the logical fragments,
without using varieties. The transfer result is generic and its proof mostly
is: the only specific argument is an \efgame game that can be adapted
to all natural fragments with minimal difficulty. The
benefits of this new proof are that:
\begin{enumerate}
\item It is self-contained and much simpler than
  previous ones. It only relies on two basic well-known notions:
  recognizability by semigroups and \efgame~games.
\item It works with classes that are not closed under complement, contrary
  to~\cite{Steinberg:delay-pointlikes:2001}. This allows us to capture the
  $\Sigma$ and $\Pi$ levels in the quantifier alternation hierarchy of
  first-order logic.
\item Under an additional hypothesis on the logical fragment, which is met for
  most fragments we investigate and easy to check, the decidability result of
  the separation problem also extends to the membership problem.
\item The proof adapts smoothly to infinite words using
  the notion of $\omega$-semigroups.
\end{enumerate}

The second form of our result is algebraic and generic. We prove that
$\mathsf{V} \mapsto \mathsf{V} \ast \mathsf{D}$ preserves the decidability of
separation for varieties, hence giving an elementary proof of a result
of~\cite{Steinberg:delay-pointlikes:2001}. Even in this algebraic form, we
completely bypass involved constructions or notions, such as pointlike sets
for categories developed in~\cite{Steinberg:delay-pointlikes:2001}, thus
making the proof~accessible.

As corollaries, since $\bso 1$ and $\sio2$ both enjoy decidable
separation~\cite{sep_icalp13,DBLP:conf/mfcs/PlaceRZ13,PZ:icalp14}, we obtain
that this is also the case for the fragments $\bsp1$ and  $\mbox{\sdp}$, known as
levels 1 and 3/2 of the dot-depth hierarchy. These new results strengthen the
previous ones~\cite{Knast:dd1:1983a,glasser-dd3/2} that showed
decidability of membership and were considered as difficult.
We actually obtain that separation for $\sip n$ reduces to separation for~$\sio n$.
Since we also transfer decidability of the membership problem, and
since  the fragments $\bso2$ of Boolean combinations of $\sio2$
formulas and $\sio3$ have decidable membership~\cite{PZ:icalp14} we deduce
that the same holds for $\bdp$ and $\stp$, known as levels 2 and 5/2 of
the dot-depth hierarchy.

\highlight{Organization of the Paper} In Section~\ref{sec:prelims}, we set up
the notation and we present the separation problem and the logics we deal
with. In Section~\ref{sec:contrib-overview}, we present an overview of our main
contribution. Section~\ref{sec:wfwords} is devoted to our technical tool:
languages of well-formed words. In Section~\ref{sec:special}, we use it to
prove our transfer result for all fragments from the logical perspective. In
Section~\ref{sec:algebra}, we establish that decidability of the separation
problem for the variety \Vbf entails the same for $\Vbf*\Dbf$. In order to
instantiate this result for concrete logical fragments, thus obtaining an
alternate proof of our transfer result, we rely on algebraic properties from
the bibliography for each fragment and its enrichment: they are presented in
Section~\ref{sec:algebraic-charac}.
% Finally, in Section~\ref{sec:infinite-words}, we describe how to extend
% these results in the context of languages of infinite words.
This paper is the full version of~\cite{PZ:plusone-stacs}.

\section{Preliminaries}
\label{sec:prelims}
In this section, we provide preliminary definitions on regular
languages defined by logical fragments and on separation.

\highlight{Words, Languages} We fix a finite alphabet $A$. Let
$A^+$ be the set of all nonempty finite words and let $A^{*}$ be the set
of all finite words over $A$. If $u,v$ are words, we denote by $u \cdot v$ or by $uv$
the word obtained by concatenating $u$ and $v$. For convenience, we
only consider, without loss of generality, languages that do not contain the empty
word. That is, a language is a subset of~$A^+$. We work with regular languages,
that is, languages definable by finite~automata.

\highlight{Separation}
Given three languages $K,L,L'$, we say that $K$ \emph{separates}
$L$ from $L'$ if
\begin{equation*}
  L \subseteq K \text{ and } K \cap L' = \emptyset.
\end{equation*}
If \Cs is a class of languages, we say that $L$ is \emph{\Cs-separable} from
$L'$ if there exists $K \in \Cs$ that separates $L$ from $L'$. Note that if
\Cs is closed under complement, $L$ is \Cs-separable from $L'$ if and only if $L'$ is
\Cs-separable from $L$. However, this is not true for a class \Cs not closed
under complement, such as the classes $\sio n$ of the quantifier alternation
hierarchy, which we shall~consider.

Given a class \Cs, the \emph{\Cs-separation problem} asks for an algorithm
which, given as input two regular languages $L,L'$, decides whether $L$ is
\Cs-separable from~$L'$. The \Cs-membership problem, which asks whether
an input regular language belongs to \Cs, reduces to the \Cs-separation
problem, as a regular language belongs to \Cs iff\ it is \Cs-separable
from its~complement.

\highlight{Logics} We investigate several fragments of
first-order logic on finite words. We view a finite word as a logical
structure made of a sequence of positions labeled over $A$. We work with
first-order logic \fow using a unary predicate $P_a$ for each $a \in A$, which
selects positions labeled with an $a$, as well as binary predicates `$=$' for
equality and `$<$' for the linear order. Such a formula defines the regular
language of all words that satisfy it. We will freely use the name of a
logical fragment of \fow to denote the class of languages definable in this
fragment. Observe that \fow is powerful enough to express the following
logical relations:
\begin{itemize}
\item First position, $min(x)$:$\qquad \forall y\ \neg (y < x)$.
\item Last position, $max(x)$:$\hspace*{4.75ex} \forall y\ \neg (x < y)$.
\item Successor, $y = x+1$:$\qquad\ \ \, x < y \wedge \neg (\exists z\ x< z \wedge z <y)$.
\end{itemize}

However, for most fragments of \fow this is not the case. For example,
in the two-variables restriction \fod of \fow, it is not possible to
express successor, as it requires quantifying over a third
variable. For these fragments $\Fs$, adding the predicates $min$,
$max$ and $+1$ yields a strictly more powerful logic $\Fs^+$. Our goal
is to prove a transfer result for such fragments: given a fragment, if
the separation problem is decidable for the weak variant~$\Fs$, then it
is decidable as well for the strong variant $\Fs^+$ obtained by
enriching $\Fs$ with the above relations. The technique is
\emph{generic}, meaning that it is not bound to a particular logic. In
particular, our transfer result applies to the following well-known
logical fragments:

\begin{itemize}[itemsep=1mm,topsep=1.5mm]
\item \foeq, the restriction of \fow in which the linear order cannot be used,
  and only equality between two positions can be tested. The enriched fragment
  $\foeqp$ ($min$ and $max$ can be eliminated from the formulas) defines locally threshold
  testable languages~\cite{Thom82}.

\item All levels in the quantifier alternation hierarchy of
  first-order logic. A first-order formula is $\sio n$ (resp.~$\pio n$) if
  its prenex normal form contains at most $(n-1)$ quantifier
  alternations and starts with an $\exists$ (resp.~a~$\forall$)
  quantifier block. Finally, a $\bso n$ formula is a boolean combination of
  $\sio n$ and $\pio n$ formulas.

  Since for all fragments above \sdo, a formula involving $min$ and $max$ can be
  expressed without these predicates in the same logic, we shall denote the enriched fragments
  by {\sip 1}, \mbox{{\bsp 1}}, and then by $\mbox{\sdp},\ \bdp$, \dots

\item \fod, the restriction of \fow using only two reusable
  variables. The corresponding enriched fragment is \fodp, since
  $min$ and $max$ can again be eliminated from the formulas.
\end{itemize}

\noindent
\figurename~\ref{fig:frag} summarizes all fragments the technique
applies to.
\begin{figure}[ht]
  \begin{center}
    \begin{tabular}{|c|c|c|c|c|}
      \hline
      Weak   variant & \foeq & \fod & \sio{n} & \bso{n} \\
      \hline
      Strong variant & \foeqp & \fodp & \sip{n} & \bsp{n}\\
      \hline
    \end{tabular}
  \end{center}
  \caption{Logical fragments to which the technique applies.}
  \label{fig:frag}
\end{figure}

\section{Overview of the Main Result}
\label{sec:contrib-overview}
In this short section, we explain our main contribution. We prove the
following~result.

\begin{theorem} \label{thm:transfer}
  Let \Fs and $\Fs^+$ be respectively the weak and strong variants of
  one of the logical fragments in \figurename~\ref{fig:frag}. Then
  $\Fs^{+}$-separability can be effectively reduced to
  $\Fs$-separability.
\end{theorem}
We actually establish two versions of this theorem:
\begin{itemize}
\item The first form, Theorem~\ref{thm:sdo}, is obtained by purely logical
  means. It is not entirely generic, since one of the directions of the
  reduction proof relies on \efgame games adapted to the fragment under
  consideration. On the other hand, it has the advantage of having a direct,
  self-contained and elementary proof, built on a constructive reduction: from
  two regular languages, we effectively build two new regular languages, and
  we exhibit an $\Fs^+$ separator for the original languages from an $\Fs$
  separator for the new~ones.

\item The second form, Theorem~\ref{thm:main}, is based on algebraic
  tools. The transfer result in this statement is presented on classes of
  finite ordered monoids or semigroups associated to the weak and enriched
  fragments respectively, through Eilenberg's correspondence. It has the
  advantage of being completely generic: no hypothesis on the algebraic class
  is assumed. Even if this approach requires some vocabulary and machinery
  from algebra, its presentation is still much simpler than the previous
  one~\cite{Steinberg:delay-pointlikes:2001}. An issue however is that, in
  order to apply this theorem to a specific fragment, one has to find
  beforehand which algebraic classes correspond to the weak and enriched
  fragments. In other terms, the statement indeed isolates a generic transfer
  property, but it relies on specific correspondences in order to be
  instantiated on a given fragment. Fortunately, the correspondences we need
  for treating all classes of \figurename~\ref{fig:frag} have already been
  established. They will be recalled in Section~\ref{sec:algebraic-charac}.
\end{itemize}

All logical fragments from \figurename~\ref{fig:frag} have a rich history
and have been extensively studied in the literature. In particular, the
separation problem is known to be decidable for the following fragments:
\foeq, \fod, {\sio 1}, {\bso 1}, {\sio
  2}~\cite{sep_icalp13,DBLP:conf/mfcs/PlaceRZ13,PZ:icalp14}. This means that,
from our results, we obtain decidability of separation for \foeqp,
\mbox{\fodp}, {\sip 1}, {\bsp 1} and \sdp.

Note that for \foeqp, \fodp and $\bsp 1$, the results could already be
obtained as corollaries of algebraic theorems of
Steinberg~\cite{Steinberg:delay-pointlikes:2001} and
Almeida~\cite{MR1709911}. As explained above, an issue with this approach is
that the proof of Steinberg's result relies on deep algebraic arguments and is
\emph{a priori} not tailored to separation: the connection with separation is
made by Almeida~\cite{MR1709911}.

For {\sip 1} and \sdp, the result is new, as Steinberg's result does not apply
to classes of languages that are not closed under complement. % Finally, we
% shall explain in Section~\ref{sec:infinite-words} how our results extend in
% the context of infinite words. These results on infinite words are new as well,
% for all considered classes.

\section{Tools for the Logical Approach: Semigroups, Well-Formed Words}
\label{sec:wfwords}
In this section, we define the main tools used for the logical approach in
this paper. 
\begin{itemize}
\item We first recall the well-known semigroup based definition of regular
  languages: a language is regular if and only if it can be recognized by a
  finite semigroup. 

\item Our second tool, \emph{well-formed words}, is specific to our problem and
  plays a key role in our transfer result. It is presented in Section~\ref{sec:well-formed-words}.
\end{itemize}
The tools specific to the algebraic approach are postponed to Section~
\ref{sec:tools-algebr-appr}.

\subsection{Semigroups and Monoids}
\label{sec:semigroups-morphisms}

We work with the algebraic representation of regular languages. Here we
briefly recall the main definitions. We refer the reader to~\cite{Pin15:MPRI}
for additional details.

\highlight{Semigroups}
A \emph{semigroup} is a set $S$ equipped
with an associative product, written $s \cdot t$ or $st$. A \emph{monoid} is a
semigroup $S$ having a neutral element $1_S$, \emph{i.e.}, such that
$s \cdot 1_S = 1_S \cdot s = s$ for all $s \in S$. If $S$ is a semigroup,
then $S^1$ denotes the monoid $S \cup \{1_S\}$ where $1_S \notin S$ is a new
element, acting as neutral element. Note that we add such a new identity even
if $S$ is already a~monoid. A semigroup morphism is a mapping $\alpha:S\to T$
from one semigroup to another which respects the algebraic structure: for all
$s,s'\in S$, we have $\alpha(s\cdot s')=\alpha(s)\cdot\alpha(s')$.
For a monoid morphism, we require additionally $S$ and $T$ to be monoids and $\alpha(1_S)=1_T$.

An element $e \in S$ is \emph{idempotent} if $e\cdot e=e$. We denote
by $E(S)$ the set of idempotents of~$S$. Given a \emph{finite}
semigroup $S$, it is folklore and easy to see that there is an integer
$\omega(S)$ (denoted by $\omega$ when $S$ is understood) such that for
all $s$ of $S$, $s^\omega$ is idempotent:~$s^\omega=s^\omega s^\omega$.

Note that $A^+$ and $A^*$ equipped with concatenation are respectively a
semigroup and a monoid called the \emph{free semigroup over $A$} and the
\emph{free monoid over $A$}. Let $L \subseteq A^+$ be a language and $S$ be a
semigroup (resp.\ a monoid). We say that $L$ is \emph{recognized by $S$} if there exist
a morphism $\alpha: A^+ \rightarrow S$ (resp. $\alpha: A^* \rightarrow S$) and
a set $F \subseteq S$ such that~$L = \alpha^{-1}(F)$.

\highlight{Semigroups and Separation} The separation problem takes as input
two regular languages $L,L'$. It is convenient to work with a single object
recognizing both of them, rather than having to deal with two. Let $S,S'$ be
semigroups recognizing $L,L'$ together with the associated morphisms
$\alpha, \alpha'$, respectively. Clearly, $L$ and $L'$ are both recognized by
$S \times S'$ with the morphism $\alpha\times \alpha':A^+\to S \times S'$
mapping $w$ to $(\alpha(w),\alpha'(w))$. From now on, we work with such a
single semigroup recognizing both languages. Replacing $S\times S'$ with its
image under $\alpha\times \alpha'$, one can also assume that this morphism is
surjective. To sum up, we assume from now on, without loss of generality, that
$L$ and $L'$ are recognized by a single surjective morphism.

\subsection{Well-Formed Words}
\label{sec:well-formed-words}

In this section, we define our main tool for this paper. Assume that
\Fs is the weak variant of one of the logical fragments of
\figurename~\ref{fig:frag} and let $\Fs^+$ be the corresponding enriched
variant. To any semigroup morphism $\alpha: A^+\to S$ into a finite
semigroup $S$, we associate a new alphabet~$\wfA$ called the alphabet
of \emph{well-formed words}. The main intuition behind this notion is
that the $\Fs^+$-separation problem for any two regular languages recognized
by $\alpha$ can be reduced to the \Fs-separation problem for two
regular languages over $\wfA$.

\medskip\noindent
The alphabet \wfA, called \emph{alphabet of well-formed words of $\alpha$}, is
defined from $\alpha: A^+\to S$ by:
\[
\wfA = (E(S) \times S \times E(S)) ~~\cup~~ (S \times E(S)) ~~\cup~~ (E(S)
\times S) ~~\cup~~ S.
\]
We will not be interested in all words of $\wfA^+$, but only in those
that are well-formed. A word $\wbb \in \wfA^+$ is said to be
\emph{well-formed} if one of the following two properties holds:
\begin{itemize}
\item \wbb is a single letter $s \in S$,
\item \wbb has length $\geqslant 2$ and is of the form
  \[
  (s_0,f_0)
  \cdot (e_1,s_1,f_1) \cdots (e_n,s_n,f_n) \cdot (e_{n+1},s_{n+1})
  \in (S \times E(S)) \cdot (E(S) \times S \times E(S))^* \cdot (E(S)
  \times S)
  \]
  with $f_i=e_{i+1}$ for all $0 \leqslant i \leqslant n$.
\end{itemize}

\begin{fact} \label{fct:reg}
  The set of well-formed words of $\wfA^+$ is a regular language.
\end{fact}

We now define a morphism $\beta: \wfA^+
\rightarrow S$ as follows. If $s \in S$, we set $\beta(s) = s$, if
$(e,s) \in E(S) \times S$, we set $\beta((e,s)) = es$, if $(s,e) \in S
\times E(S)$, we set $\beta((s,e)) = se$ and if $(e,s,f) \in \mbox{$E(S) \times
  S \times E(S)$}$, we set $\beta((e,s,f)) = esf$.

\highlight{Associated Language of Well-formed Words} To any
language $L \subseteq A^+$ that is recognized by a morphism $\alpha:A^+\to S$ into a
finite semigroup~$S$, one
associates a language of well-formed words $\Lbb \subseteq \wfA^+$:
\[
\Lbb =\bigl \{\wbb \in \wfA^+ \mid \wbb \text{ is well-formed and } \beta(\wbb)
\in \alpha(L)\bigr\}.
\]
By definition, the language $\Lbb \subseteq \wfA^+$ is the intersection
of the language of well-formed words with $\beta^{-1}(\alpha(L))$. Therefore,
it is immediate by Fact~\ref{fct:reg} that it is regular, more precisely:

\begin{fact} \label{fct:reg2} Let $L \subseteq A^+$ be a language recognized
  by a morphism $\alpha$ into a finite semigroup. Then, the associated
  language of well-formed words $\Lbb \subseteq \wfA^+$ is a regular language
  that one can effectively compute from a recognizer of~$L$.
\end{fact}

\section{Logical Approach}
\label{sec:special}
In this section, we prove Theorem~\ref{thm:transfer} from a logical
perspective. We begin with presenting our \emph{separation} theorem, which
will entail the \emph{membership} theorem as a simple consequence.

\begin{theorem} \label{thm:sdo}
  Let \Fs and $\Fs^+$ be respectively the weak and strong variants of
  one of the logical fragments in \figurename~\ref{fig:frag}.

  Let $L,L'$ be two languages recognized by a morphism $\alpha: A^+
  \rightarrow S$ into a finite semigroup~$S$. Let $\Lbb,
  \Lbb'\subseteq\wfA^+$ be the languages of well-formed words associated
  with $L,L'$, respectively. Then $L$ is $\Fs^+$-separable from $L'$
  iff\/ $\Lbb$ is $\Fs$-separable from $\Lbb'$.
\end{theorem}

Theorem~\ref{thm:sdo} reduces $\Fs^+$-separation to \Fs-separation. The latter
was already known to be decidable for several weak variants in
\figurename~\ref{fig:frag}, namely for \foeq~\cite{ltltt:2013},
\fod~\cite{DBLP:conf/mfcs/PlaceRZ13}, $\sio 1$~\cite{sep_icalp13},
$\bso 1$~\cite{sep_icalp13,DBLP:conf/mfcs/PlaceRZ13} and
\sdo~\cite{PZ:icalp14}. Hence, we get the following corollary.

\begin{corollary} \label{cor:sdp}
  Let $L,L'$ be regular languages. Then the following problems are
  decidable:
  \begin{itemize}
  \item whether $L$ is \foeqp-separable from $L'$.
  \item whether $L$ is \fodp-separable from $L'$.
  \item whether $L$ is $\sip 1$-separable from $L'$.
  \item whether $L$ is $\bsp 1$-separable from $L'$.
  \item whether $L$ is \sdp-separable from $L'$.
  \end{itemize}
\end{corollary}

Notice that since the membership problem reduces to the separation problem,
this also gives a new proof that all these fragments have a decidable
membership problem. This is of particular interest for $\fodp$, $\bsp 1$
and \sdp for which the previous proofs, which can be found in, or derived from~\cite{Str85,Almeida:1996c,PSDAD},
\cite{Knast:dd1:1983a}, and \cite{glasser-dd3/2,pw:wreath,pwdelta} respectively, are known to be quite
involved. It turns out that for \sdp, we can do even better and entirely avoid
separation. Indeed, when \Fs is expressive enough, Theorem~\ref{thm:sdo} can be
used to prove a similar theorem for the membership~problem.

\begin{theorem} \label{thm:memb}
  Let \Fs and $\Fs^+$ be respectively the weak and strong variants of
  one of the logical fragments in \figurename~\ref{fig:frag}. Moreover,
  assume that for any alphabet of well-formed words, the set of
  well-formed words over this alphabet is definable in \Fs.

  Let $L$ be a language recognized by a morphism $\alpha: A^+ \rightarrow S$ into
  a finite semigroup~$S$. Let $\Lbb \subseteq\wfA^+$ be the language of
  well-formed words associated with $L$. Then $L$ is definable in $\Fs^+$ iff\/
  $\Lbb$ is definable in $\Fs$.
\end{theorem}

\begin{proof}
  Set $K = A^+ \setminus L$ and let \Kbb be the associated language of
  well-formed words. Observe that by definition, $\Kbb \cup \Lbb$ is the
  set of all well-formed words.

  If $\Lbb$ is definable in $\Fs$, then $\Lbb$ is $\Fs$-separable from \Kbb,
  hence by Theorem~\ref{thm:sdo}, $L$ is $\Fs^+$-separable from $K$, and so
  $L$ is definable in $\Fs^+$. Conversely, if $L$ is definable in
  $\Fs^+$, then $L$ is $\Fs^+$-separable from $K$ and by Theorem~\ref{thm:sdo},
  \Lbb is $\Fs$-separable from \Kbb. Since $\Kbb \cup \Lbb$ is
  the set of all well-formed words, \Lbb is the intersection of the
  separator with the set of all well-formed words, which by hypothesis is also definable in \Fs.
  Therefore, \Lbb is definable in~\Fs.
\end{proof}

Observe that being well-formed  can be expressed in \pio{2}:
essentially, a word is well-formed if for all pairs of positions,
either there is a third one in-between, or the labels of the two
positions are ``compatible''. Hence, among the fragments of
\figurename~\ref{fig:frag}, Theorem~\ref{thm:memb} applies to all
fragments including and above \pio{2} in the quantifier alternation
hierarchy. While such a transfer result was previously
known~\cite{Str85,pw:wreath}, the presentation and the proof are
new. In particular, since membership is known to be decidable for
\pio{2}~\cite{pwdelta}, \bso{2}~\cite{PZ:icalp14} and
\sio{3}~\cite{PZ:icalp14}, we obtain new and simpler proofs of the
following results.

\begin{corollary} \label{cor:memb}
  Given a regular language $L$, one can decide whether
  \begin{itemize}
  \item $L$ is definable by a \sdp (resp. by a \pdp) formula.
  \item $L$ is definable by a \bdp formula.
  \item $L$ is definable by a \stp (resp. by a \ptp) formula.
  \end{itemize}
\end{corollary}

It remains to prove Theorem~\ref{thm:sdo}. We devote the rest of the section
to this proof. An important remark is that the proof of the right to left
direction, presented in Section~\ref{sec:from-fod-to-fodp}, is constructive:
we start with an $\Fs$ formula that separates $\Lbb$ from $\Lbb'$ and use it
to construct an $\Fs^+$ formula that separates $L$ from $L'$. Note that the
argument is generic for all fragments we consider.

On the other hand, the converse direction to which
Section~\ref{sec:from-fodp-to-fod} is devoted, namely
Proposition~\ref{prop:corr} below, requires a specific argument tailored to
each fragment: a straightforward but tedious \efgame argument.

\subsection{\texorpdfstring{From $\Fs$-separation to $\Fs^+$-separation}{From F-separation to F+-separation}}
\label{sec:from-fod-to-fodp}

We now prove that if $\Lbb$ is $\Fs$-separable from $\Lbb'$, then $L$ is
$\Fs^+$-separable from~$L'$. We do so by building an $\Fs^+$-definable
separator. This proof is constructive and entirely generic. We rely on a construction
that associates to any word $w \in A^+$
a canonical well-formed word $\croch{w} \in \wfA^+$.

\highlight{Canonical Well-formed Word Associated to a Word} To any word $w$
of $A^+$, we associate a canonical well-formed word $\croch{w} \in
\wfA^+$ such that $\alpha(w) = \beta(\croch{w})$. This construction is
adapted from~\cite{PSDAD} and is originally inspired by \cite{Str85}.

Fix an arbitrary order on the set $E(S)$.
For a position $x$ of $w$, let $u_x \in A^+$ be the infix of $w$ obtained by
keeping only positions $x-(|S|-1)$ to $x$. If position $x-(|S|-1)$ does not
exist, $u_x$ is just the prefix of $w$ ending at $x$. A position $x$ is said
\emph{distinguished} if there exists an idempotent $e \in E(S)$ such that
$\alpha(u_x) \cdot e = \alpha(u_x)$. Additionally, we always define the
rightmost position as distinguished, even if it does not satisfy the
property. Set $x_1<\cdots <x_{n+1}$ as the distinguished~positions in~$w$, so
that $x_{n+1}$ is the rightmost position. Let $e_1,\dots,e_{n}\in E(S)$ be such
that for all $1\leqslant i\leqslant n-1$, $e_i$ is the smallest idempotent such that
$\alpha(u_{x_{i}}) \cdot e_i = \alpha(u_{x_{i}})$.

If $n=0$, \emph{i.e.}, if the only distinguished position is the rightmost one,
set $\croch{w} = \alpha(w) \in \wfA$. Otherwise, we define $\croch{w} \in
\wfA^+$ as the word:
\begin{equation}
  \label{eq:croch}
  \croch{w} = (\alpha(w_0),e_1)\cdot (e_1,\alpha(w_1),e_2) \cdots
  (e_{n-1},\alpha(w_{n-1}),e_n)\cdot (e_{n},\alpha(w_{n}))
\end{equation}
\noindent where $w_0$ is the prefix of $w$ ending at position $x_1$, for all $1 \leqslant i
\leqslant n-1$, $w_i$ is the infix of $w$ obtained by keeping positions $x_{i}+1$
to $x_{i+1}$, and $w_{n}$ is the suffix of $w$ starting at position
$x_n+1$. Note that by construction, $\croch{w}$ is well-formed.

The next statement follows from the definition of $\beta$, and from the fact
that by definition of the words $w_i$ and of the chosen idempotents, we have
$\alpha(w_0\cdots w_{i})e_{i+1}=\alpha(w_0\cdots w_{i})$.

\begin{fact} \label{fct:sametype}
  For all $w \in A^+$, we have $\alpha(w) = \beta(\croch{w})$. Therefore, $w \in L$
  iff $\croch{w}\in\Lbb$ and $w\in L'$ iff $\croch{w}\in \Lbb'$.
\end{fact}

To any distinguished position $x_i$ in $w$, we now
associate the position $\croch{x} = i$ in $\croch{w}$. Our main
motivation for using this construction is its local canonicity, which is stated
in the following lemma.

\begin{lemma} \label{lem:canonic}
  Let $w \in A^+$. Then we have the following
  properties:
  \begin{enumerate}[label=$(\alph*)$,ref=$(\alph*)$]
  \item\label{item:5} whether a position $x$ is distinguished in~$w$, and if so the label of position $\croch{x}$
    in $\croch{w}$ only depends on the infix of~$w$ of length $2|S|$
    ending at position $x$. That is, if the infixes of length $2|S|$ ending at
    $x$ and $y$ are equal, then $x$ is distinguished iff so is $y$, and in that case,
    the labels of $\croch{x}$ and $\croch{y}$ in $\croch{w}$ are equal.
  \item the label of the last position of \croch{w} only depends on
    the suffix of length $2|S|$ of $w$.
  \end{enumerate}
\end{lemma}

\begin{proof}
  It is immediate that whether $x$ is distinguished and if so the associated
  idempotent only depends on the infix $u_x$ of length at most $|S|$ ending at
  $x$. Therefore, to prove~\ref{item:5}, it suffices to show that all infixes
  $w_i$ used in~\eqref{eq:croch} are of size at most $|S|$, or in other words, that among $|S| + 1$
  consecutive positions, at least one is distinguished. So let us consider an
  infix $a_1 \cdots a_{|S|+1}$ of $w$ of length $|S| + 1$. It is immediate
  from the pigeonhole principle that there exist $i<j$ such that $\alpha(a_1
  \cdots a_{i}) = \alpha(a_1 \cdots a_j) = \alpha(a_1 \cdots a_{i}) \cdot
  (\alpha(a_{i+1} \cdots a_{j}))^\omega$. Hence, the position corresponding to
  $a_i$ is distinguished. The proof of the second assertion is similar.
\end{proof}

\noindent
{\bf $L$ is $\Fs^+$-separable from $L'$.} We can now construct our
separator. The construction follows from the next proposition.

\begin{proposition} \label{prop:comp}
  Let $\Kbb \subseteq \wfA^+$ that can be defined using an $\Fs$ formula
  $\varphi$. Then there exists an $\Fs^+$ formula $\Psi$ over alphabet
  $A$ such that for every word $w \in A^+$:
  \[
  w \models \Psi \text{ if and only if } \croch{w} \models \varphi.
  \]
\end{proposition}

\begin{proof}
  Proposition~\ref{prop:comp} follows from the following simple
  consequence of Lemma~\ref{lem:canonic}.

  \begin{clm} \label{clm:canonic} For any $\abb \in \wfA$ there exists a formula
    $\gamma_\abb(x)$ of $\Fs^+$ with a free variable~$x$, such that for any $w
    \in A^+$ and any position $x$ of $w$, we have $w \models \gamma_\abb(x)$ iff
    $x$ is distinguished and $\croch{x}$ has label $\abb$ in $\croch{w}$.
  \end{clm}

  This claim holds since by Lemma~\ref{lem:canonic}, formula $\gamma_\abb(x)$ only needs to
  explore the neighborhood of size $2|S|$ of $x$, which is trivially
  possible for all fragments $\Fs^+$ we consider.
  To conclude the proof of Proposition~\ref{prop:comp}, it suffices to
  define $\Psi$ as the formula constructed from $\varphi$ by restricting
  all quantifiers to positions that are distinguished and to replace
  all tests $P_\abb(x)$ by~$\gamma_\abb(x)$.
\end{proof}

We can now finish the proof of Theorem~\ref{thm:sdo}. Assume that
$\Lbb$ is $\Fs$-separable from $\Lbb'$ and let $\varphi$ be an $\Fs$
formula defining a separator. We denote by $\Psi$ the $\Fs^+$ formula
obtained from $\varphi$ as defined in Proposition~\ref{prop:comp}. We
prove that $\Psi$ defines a language separating $L$ from $L'$.

We first prove that $L \subseteq \{w \mid w \models \Psi\}$. Assume that $w
\in L$. Then by Fact~\ref{fct:sametype}, we have $\croch{w} \in
\Lbb$. Hence, $\croch{w} \models \varphi$ and so $w \models \Psi$
by definition of $\Psi$. The proof that $L' \subseteq \{w
\mid w \not\models \Psi\}$ is identical: if  $w \in L'$,
we have $\croch{w} \in \Lbb'$ by Fact~\ref{fct:sametype}. Hence,
$\croch{w} \not\models \varphi$ and $w \not\models \Psi$ by definition
of $\Psi$.\qed

\subsection{\texorpdfstring{From $\Fs^+$-separation to $\Fs$-separation}{From F+-separation to F-separation}}
\label{sec:from-fodp-to-fod}

To complete the proof of Theorem~\ref{thm:sdo}, it remains to prove that if
$L$ is $\Fs^+$-separable from $L'$, then $\Lbb$ is $\Fs$-separable from
$\Lbb'$. The proof is this time specific to each fragment, as it requires, in
one direction of the reduction, a dedicated (but simple) \efgame argument. We
actually prove the contrapositive: if $\Lbb$ is \emph{not} $\Fs$-separable
from $\Lbb'$, then $L$ is \emph{not} $\Fs^+$-separable from~$L'$. We rely on a
construction that is dual to the one used previously: to any well-formed word
$\ubb \in \wfA^+$ and any integer $i>0$, we associate a canonical word
$\ucroch{\ubb}_i \in A^+$.

\highlight{Canonical Word Associated to a Well-formed Word} To any $s \in
S$, we associate an arbitrarily chosen nonempty word $\ucroch{s} \in A^+$
such that $\alpha(\ucroch{s})=s$ (which is possible since $\alpha$ has been
chosen surjective). Let $i>0$. From a well-formed word $\ubb \in \wfA^+$, we
build a word $\ucroch{\ubb}_i \in A^+$ as follows. If $\ubb = s \in S$, then
$\ucroch{\ubb}_i = \ucroch{s}$ for all $i$. Otherwise, we have by definition
\[
\ubb = (s_0,e_1)(e_1,s_1,e_2)\cdots (e_{n-1}s_{n-1}e_n)(e_n,s_n).
\]
For a natural $i > 0$, we set
\[
\ucroch{\ubb}_i = \ucroch{s_0}\ucroch{e_1}^{i}\ucroch{s_1}\ucroch{e_2}^{i} \cdots \ucroch{e_{n-1}}^i\ucroch{s_{n-1}}\ucroch{e_n}^{i}\ucroch{s_n}.
\]
Recall that $\beta$ is the morphism $\beta: \wfA^+ \rightarrow
S$ mapping $\ubb$ to $s_0e_1s_1\cdots s_{n-1}e_ns_n$. Since $e_j \in
E(S)$ for all $j$, it is immediate that $\alpha(\ucroch{\ubb}_i) =
\beta(\ubb)$, hence we get the following fact:
\begin{fact} \label{fct:cons1}
  For all $i > 0$ and all well-formed $\ubb \in \wfA^+$, we have $\ubb \in
  \Lbb$ (resp. $\in \Lbb'$) if and only if $\ucroch{\ubb}_i \in L$ (resp
  $\in L'$).
\end{fact}

We now proceed with the proof. We use the classical preorders associated to
fragments of first-order logic. The \emph{(quantifier) rank} $\rk\varphi$ of a first-order
formula $\varphi$ is the largest number of quantifiers along a branch in the
parse tree of~$\varphi$. Formally, $\rk\varphi=0$ if $\varphi$ is an atomic
formula, $\rk{\neg\varphi}=\rk\varphi$,
$\rk{\varphi_1\lor\varphi_2}=\max(\rk{\varphi_1},\rk{\varphi_2})$ and $\rk{\exists x\,\varphi}=\rk\varphi+1$.

 Given $u,v \in A^+$, we write $u \ksieqp v$ if any
$\Fs^+$ formula of rank $k$ that is satisfied by $u$ is satisfied by $v$ as
well. Similarly, for $\ubb,\vbb\in\wfA^+$, we write $\ubb \ksieq \vbb$ if any
$\Fs$ formula of rank $k$ that is satisfied by $\ubb$ is satisfied by $\vbb$
as well. One can verify that $\ksieq$ and $\ksieqp$ are preorders, as well as
the following standard fact:
\begin{align}
  \label{eq:1}
  \begin{aligned}
    &L \subset A^+ \text{ is definable by an $\Fs^+$ formula of rank $k$}
    \text{ iff }  L = \{u' \mid \exists u \in L \text{ st. } u \ksieqp
    u'\}\\
    &\Lbb \subset \wfA^+ \text{ is definable by an $\Fs$ formula of rank $k$}
    \text{\quad iff } \Lbb = \{\ubb' \mid \exists \ubb \in \Lbb \text{ st. } \ubb \ksieq
    \ubb'\}.
  \end{aligned}
\end{align}
Note that when $\Fs$ and $\Fs^+$ are closed under complement,
then $\ksieq$ and $\ksieqp$ are actually equivalence relations. We can now
state the main proposition of this direction.

\begin{proposition} \label{prop:corr}
  For any $k \in \nat$, there exist $\ell \in \nat$ and $i \in \nat$ such
  that  for any well-formed words $\ubb,\ubb' \in \wfA^+$ satisfying
  $\ubb \sieq{\ell} \ubb'$, we have
  $\ucroch{\ubb}_{i} \ksieqp \ucroch{\ubb'}_{i}$.
\end{proposition}

Before proving Proposition~\ref{prop:corr}, we explain how to use it to
show the first direction of Theorem~\ref{thm:sdo}.
We argue by contrapositive: assume that $\Lbb$ is \emph{not} $\Fs$-separable from $\Lbb'$. By
definition this means that no language definable in \Fs separates \Lbb
from $\Lbb'$. In particular, for any $\ell$, the
language $$\{\ubb' \mid \exists \ubb \in \Lbb \text{   st.\;} \ubb
\sieq{\ell} \ubb'\},$$ which is definable in $\Fs$ by~\eqref{eq:1}, cannot be a separator. Note that this language
contains \Lbb. Hence, for all $\ell \in \nat$, there exist $\ubb \in \Lbb$ and
$\ubb' \in \Lbb'$ such that $\ubb \sieq{\ell} \ubb'$. We deduce from
Proposition~\ref{prop:corr} and Fact~\ref{fct:cons1} that for all $k \in
\nat$, there exist $u \in L$ and $u' \in L'$ such that $u \ksieqp u'$. It
follows, again by~\eqref{eq:1}, that $L$ is \emph{not} $\Fs^+$-separable
from $L'$, which concludes the~proof.

We now prove Proposition~\ref{prop:corr} for fragments we are interested
in. As already explained, this proposition is proved using classical, but
specific \efgame arguments for each fragment. While each proof is specific,
the underlying ideas are similar.

Here, we consider two main cases, $\Fs = \fod$ and $\Fs = \sio{n}$ for some
$n$. Note that we will obtain the case $\Fs = \bso{n}$ as a simple consequence
of the case $\Fs=\sio{n}$. Finally, we leave out the case $\Fs = \foeq$, as the
argument is essentially a copy and paste of the argument for \sio{n}.

\subsubsection{\texorpdfstring{\fod and \fodp}{FO2(<) and FO2(<,+1)}}
\label{sec:fod-fodp-1}
Observe that since \fod and \fodp are both closed under complement,
the preorders $\ksieq$ and $\ksieqp$ are actually equivalence
relations. To avoid confusion with other fragments, we denote by
\kfodeq and \kfodeqp, these two equivalences. We prove
the following proposition, which clearly entails Proposition~\ref{prop:corr}.

\begin{proposition} \label{prop:fodcor}
  For any $k \in \nat$, given $\ubb,\ubb' \in \wfA^+$ we have the
  following implication:
  \[
  \ubb \kfodeq \ubb' \Rightarrow \ucroch{\ubb}_{2k} \kfodeqp
  \ucroch{\ubb'}_{2k}.
  \]
\end{proposition}

This is proved using an \efgame argument. We first define the \efgame
game associated to $\fod$ (\emph{i.e.,} corresponding to \kfodeq) and
then explain how to adapt it to $\kfodeqp$.

\highlight{\efgame Game} The board of the $\fod$-game consists of two words
and lasts a predefined number $k$ of rounds. There are two players
called \emph{Spoiler and Duplicator}. At any time during the game there
is one pebble placed on a position of one word and one pebble placed
on a position of the other word, and both positions have the same
label. When the game starts, both pebbles are placed on the first
position of each words. Each round starts with Spoiler choosing one
of the pebbles, and moving it inside its word from its original
position $x$ to a new position $y$. Duplicator must answer by moving
the other pebble in the other word from its original position $x'$ to
a new position $y'$. Moreover, $x'$ and $y'$ must satisfy the same
relations as $x$ and $y$ among '$<$' and the label predicates.

Duplicator wins if she manages to play for all $k$ rounds. Spoiler
wins as soon as Duplicator is unable to play.

The \fodp-game is defined similarly with additional constraints for
Duplicator. When Spoiler makes a move, Duplicator must choose her
answer $y'$ so that $x'$ and $y'$ satisfy the same relations as $x$
and $y$ among $+1$, $<$ and the label predicates.

\begin{lemma}[Folklore] \label{lem:efgame}
  For any integer $k$ and any words $v,v'$, we have the following facts:
  \begin{itemize}
  \item $v \kfodeq v'$ iff Duplicator has a winning strategy in the
    $k$-round \fod-game on $v$ and $v'$.
  \item $v \kfodeqp v'$ iff Duplicator has a winning strategy in the
    $k$-round \fodp-game on $v$ and~$v'$.
  \end{itemize}
\end{lemma}

To prove Proposition~\ref{prop:fodcor}, let $\ubb,\ubb'\in\wfA^+$, and set
$u = \ucroch{\ubb}_{2k}$ and $u'= \ucroch{\ubb'}_{2k}$. We want to show that
$u \kfodeqp u'$.  In view of Lemma~\ref{lem:efgame}, it is enough to prove to
exhibit a winning strategy for Duplicator in the $k$-round \fodp-game played
on $u$ and $u'$. We call \Gs this game. The strategy involves playing a shadow
\fod-game $\Ss$ on \ubb and $\ubb'$. Observe that by hypothesis and by
Lemma~\ref{lem:efgame}, Duplicator has a winning strategy for $k$ rounds in
the game \Ss. We begin by setting up some notation to help us define
Duplicator's strategy in $\Gs$.

\highlight{Notation}  Assuming that $\ubb \kfodeq \ubb'$, we need
to prove that $u \kfodeqp u'$. If $\ubb \in S$ or $\ubb' \in S$, then
$\ubb = \ubb' = s \in S$ (since the only well-formed word that contains
letter $s \in S$ is $s$ itself) and the result is immediate.

Otherwise, by hypothesis, the words $\ubb$ and $\ubb'$ are of the form
\[
\begin{array}{lcl}
  \ubb  & = & (s_0,e_1)(e_1,s_1,e_2)\cdots (e_m,s_m)\\
  \ubb' & = & (s'_0,e'_1)(e'_1,s'_1,e'_2)\cdots(e'_{m'},s'_{m'}).
\end{array}
\]
In particular, observe that since $\ubb \kfodeq \ubb'$ and the labels
of the leftmost and rightmost positions occur only at these positions in $\ubb$ and
$\ubb'$, we have $(s_0,e_1)= (s'_0,e'_1)$ and
$(e_m,s_m) = (e'_{m'},s'_{m'})$. For the sake of simplifying the presentation, we
assume that for all $i \leqslant m$, we have $\ucroch{s_i} = a_i \in A$
and $\ucroch{e_i} = b_i \in A$ (this does not harm the generality
of the proof). Similarly, for all $i \leqslant m'$, we assume that
$\ucroch{s'_i} = a'_i \in A$ and $\ucroch{e'_i} = b'_i \in A$. By
definition, we have
\[
\begin{array}{lcl}
  u     & = & \ucroch{\ubb}_{2k} =
              a_{0}(b_{1})^{2k}a_{1}(b_{2})^{2k} \cdots (b_m)^{2k}a_{m}\\
  u'    & = & \ucroch{\ubb'}_{2k} =  a'_{0}(b'_1)^{2k}a'_1(b'_2)^{2k}
              \cdots (b'_{m'})^{2k}a'_{m'}.
\end{array}
\]
To treat the beginning and the end of the words uniformly as the other
factors, we set $b^{}_0,b'_0,b^{}_{m+1},b'_{m'+1}$ as the empty word.

\highlight{Winning Strategy} Let $\ell$ be the number of remaining rounds at
some point in the game. We define an invariant $\Is(\ell)$ that Duplicator has
to satisfy when playing. Assume that the pebbles in $u,u'$ are at positions
$x,x'$ in \Gs and that the pebbles in $\ubb,\ubb'$ are at positions $i,i'$ in
\Ss. Then, $\Is(\ell)$ holds when so do all following properties:
\begin{enumerate}
\item\label{item:6} Duplicator has a winning strategy for playing $\ell$ rounds in
  \Ss. In particular, this means that $i,i'$ have the same label, and therefore
  that $(b_i,a_i,b_{i+1}) = (b'_{i'},a'_{i'},b'_{i'+1})$.
\item\label{item:7} Pebbles $x$ and $x'$ are inside the identical factors
  $(b_{i})^{2k}a_{i}(b_{i+1})^{2k}$ and
  $(b'_{i'})^{2k}a'_{i'}(b'_{i'+1})^{2k}$, and at the same relative position.
\item\label{item:8} There are at least $\ell$ copies of $b_{i}$ (resp $b'_{i'}$) to
  the left of $x$ (resp. $x'$) and $\ell$ copies of $b_{i+1}$
  (resp. $b'_{i'+1}$) to the right of $x$ (resp. $x'$).
\end{enumerate}

It is clear that $\Is(k)$ holds at the beginning of the game. Moreover, since
Duplicator will follow her strategy in \Ss, Item~\ref{item:6} will be
fulfilled. Assume now that $\Is(\ell+1)$ holds and that there are $(\ell+1)$
rounds left to play. We explain how Duplicator can answer a move by Spoiler
while enforcing $\Is(\ell)$. Assume that Spoiler moves the pebble in $u$ to a
new position~$y$ (the dual case, when Spoiler plays in $u'$, is treated
similarly). There are two distinct cases.

\begin{itemize}
\item If $y$ remains in the factor $(b_{i})^{2k}a_{i}(b_{i+1})^{2k}$
  and satisfies Item~\ref{item:8} of $\Is(\ell)$, then Duplicator simply copies
  Spoiler's move in $(b'_{i'})^{2k}a'_{i'}(b'_{i'+1})^{2k}$. The
  positions $i$ and $i'$ remain unchanged and $\Is(\ell)$ is clearly
  satisfied.

\item Otherwise, observe that by $\Is(\ell+1)$, Spoiler's move $y$ cannot be
  equal to $x\pm 1$. This means that Duplicator has to answer in the same
  direction and on the same label as Spoiler did, but not on positions
  $x'\pm1$. Because Item~\ref{item:8} is not satisfied, position $y$ belongs
  to some $(b_{j})^{2k}a_{j}(b_{j+1})^{2k}$ with $j\neq i$, with at least
  $\ell$ copies of $b_{j}$ to its left and $\ell$ copies of $b_{j+1}$ to its
  right. To compute her answer, Duplicator simulates a move by Spoiler in~\Ss
  by moving the pebble from position $i$ to position $j$. From her winning
  strategy in~\Ss, she obtains a position $j'$ in $\ubb'$ such that
  $(b_j,a_j,b_{j+1}) = (b'_{j'},a'_{j'},b'_{j'+1})$. She picks as position
  $y'$ the same relative position in $(b'_{j'})^{2k}a'_{j'}(b'_{j'+1})^{2k}$
  as $y$ in $(b_{j})^{2k}a_{j}(b_{j+1})^{2k}$. Observe that since $j\neq i$,
  we have $y\neq x\pm 1$. Hence, this is a legal move for Duplicator. The new positions $y,y',j,j'$ satisfy $\Is(\ell)$,
  which terminates the proof.
\end{itemize}

\subsubsection{\texorpdfstring{\sio{n} and \sip{n}}{Σ-n(<) and Σ-n(<,+1,min,max)}}
\label{sec:sion-sipn}
We fix some $n \in \nat$. We keep using the symbols $\ksieq$ and
$\ksieqp$ to denote the preorders associated to $\sio{n}$ and
$\sip{n}$. Furthermore, we denote by $\kbceq$ and $\kbceqp$ the
equivalence relations associated to $\bso{n}$ and $\bsp{n}$. We prove
the following proposition, which again yields Proposition~\ref{prop:corr} with
$k'=k$ and $i=2^{k+1}$.

\begin{proposition} \label{prop:sicor}
  For any $k \in \nat$, given $\ubb,\ubb' \in \wfA^+$ we have the
  following implications:
  \[
  \begin{array}{lcl}
    \ubb \ksieq \ubb' & \Rightarrow & \ucroch{\ubb}_{2^{k+1}} \ksieqp
                                      \ucroch{\ubb'}_{2^{k+1}} \\[1.5ex]
    \ubb \kbceq \ubb' & \Rightarrow & \ucroch{\ubb}_{2^{k+1}} \kbceqp
                                      \ucroch{\ubb'}_{2^{k+1}}.
  \end{array}
  \]
\end{proposition}

Observe first that the second implication is an immediate consequence
of the first one. Indeed, since $\bsc{n}$ formulas are boolean
combinations of \sic{n} formulas, we have
\[
\begin{array}{rcl}
  v \ksieq v' \text{ and } v' \ksieq v & \text {if and only if} & v \kbceq
                                                                  v' \\[1.5ex]
  v \ksieqp v' \text{ and } v' \ksieqp v & \text{if and only if} & v \kbceqp
                                                                   v'.
\end{array}
\]
Therefore, we concentrate on the first % $\ksieq$ and $\ksieqp$
implication. As for \fod, this an \efgame argument. We first define
the \efgame game associated to $\sio{n}$ (\emph{i.e.,} corresponding
to \ksieq) and then explain how to adapt it to $\ksieqp$.

\highlight{\efgame Game} The board of the $\sio{n}$-game consists of two words
$v,v'$ and there are two players, again called \emph{Spoiler and
  Duplicator}. Moreover, initially, there exists a distinguished word among
$v,v'$ that we call the \emph{active word} (this word may change as the game
progresses). The game is set to last a predefined number $k$ of rounds. When
the game starts, both players have $k$ pebbles. Contrary to the \fod-game,
once a pebble is dropped, it cannot be moved again during the game. Finally,
there is a parameter that gets updated during the game, a counter $c$ called
the \emph{alternation counter}. Initially, $c$ is set to $0$. It may be
incremented, but it has to remain bounded by $n-1$.

At the start of each round $\ell$, Spoiler chooses a word, either $v$ or
$v'$. Spoiler can always choose the active word, in which case both $c$
and the active word remain unchanged. However, Spoiler can only choose
the word that is not active when $c <n-1$, in which case the active
word is switched and $c$ is incremented by $1$ (in particular,
this may happen at most $n-1$ times). If Spoiler chooses $v$
(resp. $v'$), he puts a pebble on a position $x_\ell$ in $v$
(resp. $x'_\ell$ in~$v'$).

Duplicator must answer by putting a pebble at a position $x'_\ell$ in
$v'$ (resp. $x_\ell$ in $v$). Moreover, Duplicator must ensure that all
pebbles that have been placed up to this point verify the following
condition: for all  $\ell_1,\ell_2 \leqslant \ell$, the labels at positions
$x_{\ell_1},x'_{\ell_1}$ are the same, and $x_{\ell_1} < x_{\ell_2}$
if and only if $x'_{\ell_1} < x'_{\ell_2}$.

Duplicator wins if she manages to play for all $k$ rounds, and Spoiler
wins as soon as Duplicator is unable to play.

The \sip{n}-game is defined similarly with the following additional constraint
for Duplicator: at any time, %the pebbles must additionally satisfy:
for all $\ell_1,\ell_2\leqslant\ell$, we have $x_{\ell_1} = x_{\ell_2} + 1$ if and only if
$x'_{\ell_1} = x'_{\ell_2} + 1$, $min(x_{\ell_1})$ if and only if
$min(x'_{\ell_1})$ and $max(x_{\ell_1})$ if and only if $max(x'_{\ell_1})$.

\begin{lemma}[Folklore] \label{lem:efgame2}
  For all $k \in \nat$ and $v,v'$, we have the following facts:
  \begin{itemize}
  \item $v \ksieq v'$ iff Duplicator has a winning strategy in the
    $k$-round \sio{n}-game on $v$ and $v'$ with $v$ as initial active
    word.
  \item $v \ksieqp v'$ iff Duplicator has a winning strategy in the
    $k$-round \sip{n}-game on $v$ and~$v'$ with $v$ as initial active
    word.
  \end{itemize}
\end{lemma}

We now prove Proposition~\ref{prop:sicor}. Let $\ubb,\ubb'\in\wfA^+$. We have
to prove that $\ucroch{\ubb}_{2^{k+1}} \ksieqp \ucroch{\ubb'}_{2^{k+1}}$. In
view of Lemma~\ref{lem:efgame2}, this can be done by giving a winning strategy
for Duplicator in the corresponding $k$-round \sip{n}-game. We call \Gs this
game. Duplicator's strategy involves playing a \sio{n}-game \Ss, called the
shadow game, on \ubb and $\ubb'$. By hypothesis and by Lemma~\ref{lem:efgame2},
she has a winning strategy in $k$ rounds in the shadow game \Ss. We begin by
setting up some notation that will help us define Duplicator's strategy.

\highlight{Notation} Set $u = \ucroch{\ubb}_{2^{k+1}}$ and $u'=
\ucroch{\ubb'}_{2^{k+1}}$. Assuming that $\ubb \ksieq \ubb'$, we need
to prove that $u \ksieqp u'$. If $\ubb \in S$ or $\ubb' \in S$, then
$\ubb = \ubb' = s \in S$ (again, the only well-formed word that contains the
letter $s \in S$ is $s$). Therefore, $u=u'$ and the result is immediate.

\medskip
Otherwise, by hypothesis, the words $\ubb$ and $\ubb'$ are of the form
\[
\begin{array}{lcl}
  \ubb  & = & (s_0,e_1)(e_1,s_1,e_2)\cdots (e_m,s_m)\\
  \ubb' & = & (s'_0,e'_1)(e'_1,s'_1,e'_2)\cdots(e'_{m'},s'_{m'})
\end{array}
\]
In particular, observe that since $\ubb \ksieq \ubb'$ and the labels
of the leftmost and rightmost positions occur only at these positions in $\ubb$ and
$\ubb'$, we have  $(s_0,e_1)= (s'_0,e'_1)$ and
$(e_m,s_m) = (e'_{m'},s'_{m'})$. For the sake of simplifying the presentation, we
assume that for all $i \leqslant m$, we have $\ucroch{s_i} = a_i \in A$
and $\ucroch{e_i} = b_i \in A$ (this does not harm the generality
of the proof). Similarly, for all $i \leqslant m'$, we assume that
$\ucroch{s'_i} = a'_i \in A$ and $\ucroch{e'_i} = b'_i \in A$. By
definition, we have
\[
\begin{array}{lcl}
  u     & = & \ucroch{\ubb}_{2^{k+1}} =
              a_{0}(b_{1})^{2^{k+1}}a_{1}(b_{2})^{2^{k+1}} \cdots (b_m)^{2^{k+1}}a_{m}\\
  u'    & = & \ucroch{\ubb'}_{2^{k+1}} =  a'_{0}(b'_1)^{2^{k+1}}a'_1(b'_2)^{2^{k+1}}
              \cdots (b'_{m'})^{2^{k+1}}a'_{m'}.
\end{array}
\]
Again, to treat the beginning and the end of the words uniformly as the other
factors, we set $b^{}_0,b'_0,b^{}_{m+1},b'_{m'+1}$ as the empty word.

\highlight{Winning Strategy} Let $\ell$ be the number of remaining rounds at
some point in the game. We define an invariant $\Is(\ell)$ that Duplicator has
to satisfy when playing.

As she plays, Duplicator associates to each position $i \in \ubb$, (resp.
$i' \in \ubb'$) a set of positions in $u$ (resp. $u'$) called the set of
\emph{marked positions} for $i$ (resp.~for $i'$). All marked positions for $i$
(resp.~for $i'$) must belong to the $b_i$, $a_i$ or $b_{i+1}$
(resp.~$b'_{i'}$, $a'_{i'}$ or $b_{i'+1}$) positions in $u$ (resp.~in~$u'$). Initially, for all $i$ (resp.\ $i'$), only $a_i$ (resp. $a'_{i'}$) is
marked for $i$ (resp.~for $i'$). Duplicator may define more positions as
marked as the game progresses. All these new marked positions will be
positions holding pebbles in \Gs.

Assume that there are $\ell$ rounds left to play and that pebbles have
already been placed on $u,u'$ in the main game \Gs and on $\ubb,\ubb'$
in \Ss in a way that satisfies the conditions of both \efgame games. We
denote by $c_\Gs$ the alternation counter of the main game \Gs and by
$c_{\Ss}$ that of the shadow game $\Ss$. For all $i \in \ubb$
(resp. $i' \in \ubb'$) we denote by $x_1(i) < \cdots < x_{m_i}(i)$
(resp. $x'_1(i') < \cdots < x'_{m_{i'}}(i')$) the marked positions for $i$
(resp. $i'$). Then $\Is(\ell)$ holds if the following properties hold:
\begin{enumerate}
\item\label{item:9} Duplicator has a winning strategy for playing at least $\ell$
  more rounds in \Ss. Furthermore, either $c_{\Ss} > c_\Gs$, or $c_\Ss =
  c_{\Gs}$ and the active words in \Ss and \Gs are either $\ubb$ and
  $u$, or $\ubb'$ and~$u'$.
\item\label{item:10} Any position $x \in u$ (resp.~$x' \in u'$) that holds a pebble in \Gs
  is marked for some $i \in \ubb$ (resp. $i' \in \ubb'$) holding a
  pebble in \Ss. Conversely, any position that is marked for $i \in
  \ubb$, (resp. $i' \in \ubb'$) is either $a_i$ (resp. $a'_{i'}$) or a
  position holding a pebble in \Gs.

\item\label{item:11} For all $i \in \ubb$ (resp. $i' \in \ubb'$), we have $x_{m_i}(i) <
  x_1(i+1)$ (resp. $x'_{m_{i'}}(i') < x'_1(i'+1)$). Moreover, there
  are at least $2^{\ell+1}$ copies of $b_{i+1}$ (resp. $b'_{i'+1}$)
  that are strictly between these two positions.

\item\label{item:12} Let $i,i'$ be positions of $\ubb,\ubb'$ on which there are corresponding
  pebbles in \Ss (meaning that one position corresponds to a move of Spoiler
  and the other one is Duplicator's answer). Observe that since $i,i'$ have
  the same label, we have $a_i = a'_{i'}$, $b_i = b'_{i'}$ and $b_{i+1} = b'_{i'+1}$.
  In that case, the number of marked positions for $i$ is the same as the
  number of marked positions for $i'$, that is $m_i = m_{i'}$. Furthermore,
  for all $j \leqslant m_i$:
  \begin{itemize}
  \item $x_j(i)$ is the $a_i = a'_{i'}$ position of $u$ iff $x'_j(i')$ is
    the $a_i = a'_{i'}$ position of $u'$, and
  \item  $x_j(i)$ holds a pebble of
    \Gs iff $x'_j(i')$ holds the corresponding pebble.
  \end{itemize}

  Finally, given $j < m_i$, let $d$ and $d'$ be the number of positions that
  are strictly between $x_j(i)$ and $x_{j+1}(i)$ (resp.~between $x'_j(i')$ and
  $x'_{j+1}(i')$). Note that by the condition above these positions are all
  labeled by $b^{}_i = b'_{i'}$, or all labeled by $b^{}_{i+1} =
  b'_{i'+1}$.
  We require that either $d = d'$, or $d \geqslant 2^\ell$ and
  $d' \geqslant 2^\ell$.
\end{enumerate}

\begin{figure}[ht]
  \centering
  \begin{tikzpicture}[xscale=.8]
    \node (u) at (-2,0) {$u'$};
    \node (cu) at (-2,-1.5) {$\ubb'$};
    \node (uend) at (15,0) {};
    \node (cuend) at (15,-1.5) {};
    \node (ca1) at (0,-1.5) {}; 
    \node (ca7) at (14,-1.5) {}; 

    \tikzstyle{every node}=[draw,thin,circle,inner sep=1pt,radius=3pt, minimum size=.2cm,fill=black!50];
    \node (a1) at (0,0) {};
    \node[fill=white] (a2) at (6,0) {};
    \node (a3) at (7.2,0) {};    
    \node (a4) at (10,0) {};
    
    \node (b1) at (6,-1.5) {}; 

    \tikzstyle{every node}=[above=-3pt];
    \node at (a1.north) {$\begin{array}{c}x_1(i')\\b'_{i'}(=b_i)\end{array}$}; %{$\begin{array}{c}x_1(i')\\b_{i+1}_{i'}\end{array}$}
    \node at (a2.north) {$\begin{array}{c}x_2(i')\\a_i\end{array}$}; % {$\begin{array}{c}x_2(i')\\a'_{i'}\end{array}$}
    \node at (a3.north) {$\begin{array}{c}x_3(i')\\b_{i+1}\end{array}$}; %{$\begin{array}{c}x_3(i')\\b_{i+1}_{i'+1}\end{array}$}
    \node at (a4.north) {$\begin{array}{c}x_4(i')\\b_{i+1}\end{array}$}; %{$\begin{array}{c}x_4(i')\\b_{i+1}_{i'+1}\end{array}$}
    
    \tikzstyle{every node}=[below=-2pt];
    \node at ($(b1.south)$) {$\begin{array}{c}a_i\\i'\end{array}$};

    \draw[dashed] (b1) to (a1);
    \draw[dashed] (b1) to (a2);
    \draw[dashed] (b1) to (a3);
    \draw[dashed] (b1) to (a4);

    \draw (a1) to (a2) to (a3) to (a4) to ($(uend)-(.5,0)$) ;
    \draw [dotted] (u) to (a1);
    \draw [dotted] (a4) to (uend);

    \draw ($(ca1)-(1,0)$) to (b1) to ($(ca7)+(.5,0)$);
    \draw ($(ca1)-(1,0)+(a2)-(b1)$) to (a1);
    \draw [dotted] (cu) to (ca1);
    \draw [dotted] (ca7) to (cuend);
    % \draw[decorate,decoration=brace,very thick] ($(a4)+(0,1.2)$) -- ($(a5)+(-0.05,1.2)$);
    % \draw[decorate,decoration=brace,very thick] ($(a5)+(0.05,1.2)$) -- ($(a6)+(0,1.2)$);
    \tikzstyle{every node}=[];

    \begin{scope}[yshift=3.5cm]
    \node (u) at (-2,0) {$u$};
    \node (cu) at (-2,1.5) {$\ubb$};
    \node (uend) at (15,0) {};
    \node (cuend) at (15,1.5) {};
    \node (ca1) at (0,1.5) {}; 
    \node (ca7) at (14,1.5) {}; 

    \tikzstyle{every node}=[draw,thin,circle,inner sep=1pt,radius=3pt, minimum size=2mm,fill=black!50];
    \node (a1) at (0,0) {}; 
    \node[fill=white] (a2) at (3,0) {}; 
    \node (a3) at (4.2,0) {};    
    \node (a4) at (7,0) {};
    
    \node (a5) at (10.2,0) {}; 
    \node (a6) at (12,0) {};        
    \node (a7) at (14,0) {};

    \node (b1) at (3,1.5) {}; 
    \node (b2) at (12,1.5)  {}; 

    \tikzstyle{every node}=[below=-2pt];
    \node at (a1.south) {$\begin{array}{c}b_i\\x_1(i)\end{array}$};
    \node at (a2.south) {$\begin{array}{c}a_i\\x_2(i)\end{array}$};
    \node at (a3.south) {$\begin{array}{c}b_{i+1}\\x_3(i)\end{array}$};
    \node at (a4.south) {$\begin{array}{c}b_{i+1}\\x_4(i)\end{array}$};
    
    \node at (a5.south) {$\begin{array}{c}b_{i+1}\\x_1(i+1)\end{array}$};
    \node at (a6.south) {$\begin{array}{c}a_{i+1}\\x_2(i+1)\end{array}$};
    \node at (a7.south) {$\begin{array}{c}b_{i+2}\\x_3(i+1)\end{array}$};
    
    \tikzstyle{every node}=[above=-3pt];
    \node at (b1.north) {$\begin{array}{c}i\\a_i\end{array}$};
    \node at (b2.north) {$\begin{array}{c}{i+1}\\a_{i+1}\end{array}$};

    \draw[dashed] (b1) to (a1);
    \draw[dashed] (b1) to (a2);
    \draw[dashed] (b1) to (a3);
    \draw[dashed] (b1) to (a4);

    \draw[dashed] (b2) to (a5);
    \draw[dashed] (b2) to (a6);
    \draw[dashed] (b2) to (a7);

    \draw (a1) to (a2) to (a3) to (a4) to (a5) to (a6) to (a7) to ($(a7)+(.5,0)$);
    \draw [dotted] (u) to (a1);
    \draw [dotted] (a7) to (uend);

    \draw ($(ca1)-(1,0)$) to (b1) to (b2) to ($(ca7)+(.5,0)$);
    \draw ($(ca1)-(1,0)+(a2)-(b1)$) to (a1);
    \draw [dotted] (cu) to (ca1);
    \draw [dotted] (ca7) to (cuend);
    \draw[decorate,decoration=brace,very thick] ($(a5)-(0.05,1.2)$) -- ($(a4)-(0,1.2)$) ;
    \draw[decorate,decoration=brace,very thick] ($(a2)-(0,1.2)$) -- ($(a1)-(-0.05,1.2)$) ;
    \node at ($.5*(a4)+.5*(a5)+(0,-1.9)$) {$\geqslant 2^{\ell+1}$};
    \node at ($.5*(a1)+.5*(a2)+(0,-1.9)$) {$\geqslant 2^{\ell}$};
    \tikzstyle{every node}=[];
  \end{scope}      
 \end{tikzpicture}
  \caption{Marked positions}
  \label{fig:marked}
\end{figure}

\figurename~\ref{fig:marked} shows positions $i$ and $i+1$ in $\ubb$ and~$i'$
in $\ubb'$ corresponding to $i$ in the \Ss-game, as well as marked positions for $i$ and $i+1$ in~$u$ (resp.~for $i'$
 in~$u'$). Greyed positions are the ones holding a pebble. Note that
by Item~\ref{item:10}, all marked positions in $u$ (resp.~$u'$) except
possibly some $a_i$ (resp.~$a'_i$) positions have to hold a
pebble. Item~\ref{item:12} means that the picture for $u'$ and $\ubb'$ look
the same: for instance, since there are $m_i=4$ marked positions for $i$
in~$u$, there are also 4 marked positions for $i'$ in $u'$, where $i$ and $i'$
are corresponding moves in \Ss. Furthermore, all are marked except the
$a=a^{}_i=a'_{i'}$ position in both $u$ and $u'$, and this position has the
same index in both lists of marked positions for $i$ (resp.~$i'$), namely
index~2. Finally, distances between ``corresponding'' consecutive marked
positions in $u$ and $u'$ are either equal, or both are at least~$2^\ell$. In
\figurename~\ref{fig:marked}, $x_2(i)-x_1(i)\neq x'_2(i')-x'_1(i')$, therefore
these quantities have to be at least $2^\ell$.

\medskip
It is clear that $\Is(k)$ holds before the initial round. Assume now
that there are $(\ell+1)$ rounds left to play and that $\Is(\ell+1)$
holds. We explain how Duplicator can play in order to enforce
$\Is(\ell)$. Assume that Spoiler puts a pebble at a position $x \in
u$ in \Gs (the case when Spoiler plays in $u'$ is symmetric).

Duplicator first defines a position $i'$ in $\ubb'$ as follows. If there is
already a pebble on $i$ in \Ss, then we set $i'$ as the position holding the
matching pebble in $\ubb'$. Otherwise, Duplicator simulates a move by Spoiler
in \Ss by putting a pebble on position $i$ and sets $i'$ as the answer she
obtains from her strategy in \Ss. Note that by hypothesis all pebbles in
$\ubb,\ubb'$ (including $i,i'$) satisfy the conditions of the \sio{n}-game.
We now distinguish two cases depending on the position $x$.

\highlight{There exists $i \in \ubb$ such that
  $x_{1}(i) \leqslant x \leqslant x_{m_{i}}(i)$} 
We distinguish two subcases:
\begin{itemize}
\item If $x$ is already a marked position $x_j(i)$ for $i$, then Duplicator
  answers by putting a corresponding pebble on $x'_{j}(i')$. Note that this
  answer is correct by hypothesis on $i,i'$ for the \sio{n}-game and by
  hypothesis on the marked positions for $i,i'$ as stated in Item~\ref{item:10} of
  $\Is(\ell+1)$. Since both positions were already marked for $i,i'$, it is
  then simple to verify that $\Is(\ell)$ holds.

\item Assume now that $x$ is not yet marked. Since $a_i$ positions are always
  marked, $x$ is a $b_i$ or a $b_{i+1}$ position. Assume that $x$ is a $b_i$
  position (the other case is similar). Recall that $m_i = m_{i'}$ by
  Item~\ref{item:10} in $\Is(\ell+1)$. Let $j$ be such that $x_j(i) < x < x_{j+1}(i)$.  By
  Item~\ref{item:12} of $\Is(\ell + 1)$ it is immediate than one can find an answer
  $x' \in u'$ such that $x'_j(i') < x' < x'_{j+1}(i')$ and Item~\ref{item:12} of
  $\Is(\ell)$ remains satisfied with $x,x'$ as new marked positions for
  $i,i'$. Again this answer is correct by hypothesis on $i,i'$ for the
  \sio{n}-game and by hypothesis on the marked positions for $i,i'$ as stated
  in Item~2 of $\Is(\ell+1)$. It is then simple to verify that $\Is(\ell)$
  remains satisfied.
\end{itemize}

\highlight{There exists $i \in \ubb$ such that $x_{m_{i-1}}(i-1) < x <
  x_1(i)$} From Item~\ref{item:11} in $\Is(\ell+1)$, we know that there are at
least $2^{\ell+2}$ copies of $b_{i}$ between $x_{m_{i-1}}(i-1)$ and
$x_1(i)$. It follows that there are either at least $2^{\ell+1}$
copies of $b_{i}$ between $x_{m_{i-1}}(i-1)$ and $x$ or at least
$2^{\ell+1}$ copies of $b_{i}$ between $x$ and $x_1(i)$. Since both cases are
symmetric, assume that we are in the first case: there are at least
$2^{\ell+1}$ copies of $b_{i}$ between $x_{m_{i-1}}(i-1)$ and $x$.

\smallskip

Let $d$ be the number of copies of $b_i$ between $x$ and $x_1(i)$, \emph{i.e.},
$x = x_1(i) - (d+1)$. If $d < 2^\ell$, we set $x' \in u'$ as the position
$x' = x_1(i') - (d+1)$. Otherwise we set $x' \in u'$ as the position
$x' = x_1(i) - (2^\ell+1)$. In both cases, $x'$ is Duplicator's answer and we
set $x,x'$ as new marked positions for $i,i'$. Note that this answer is
correct by hypothesis on $i,i'$ for the \sio{n}-game. It is immediate that
$\Is(\ell)$ are satisfied by choice of $x'$.

\section{Tools for the Algebraic Approach: Varieties, Semidirect Product}
\label{sec:tools-algebr-appr}

In this section, we set up the terminology needed for the algebraic version of
our result. As explained in the introduction, we use varieties to capture our
classes of separator languages. Informally, a variety is a class of finite
algebras canonically associated to such a class of separators. We build our
algebraic version of the transfer theorem from a weak fragment \Fs to its enriched
version $\Fs^+$ on three ingredients:
\begin{enumerate}[label=\itemfmt{I\arabic*.},ref=\itemref{I\arabic*}]
\item\label{item:13} A solution to the separation problem for \Fs, as in the logical approach.
\item\label{item:14} An algebraic description of the weak variant \Fs as a variety % of monoids
  \Vbf.
\item\label{item:15} An algebraic description of the strong variant $\Fs^+$ as the variety % of
  % semigroups
  $\Vbf * \Dbf$, built from $\Vbf$ and from a fixed variety $\Dbf$ with an
  operator called the semidirect product.
\end{enumerate}
These three points have already been solved for all fragments of
\figurename~\ref{fig:frag}. The transfer result, Theorem~\ref{thm:main} below,
reduces separability by languages associated with $\Vbf * \Dbf$ to
separability by languages associated with $\Vbf$. Therefore, relying on the
solution of Items~\ref{item:14} and \ref{item:15}, it provides a
reduction from the separation problem by~\Fs languages to
the separation problem by $\Fs^+$~languages. If in addition Item~\ref{item:13}
if fulfilled, then the latter problem is decidable.

\medskip
This section is devoted to making these notions precise. It is organized as
follows: we first recall the notion of variety of ordered semigroups and
monoids, and how varieties can be used to capture classes of regular languages
we are interested in. We then recall the construction of the semidirect
product of two varieties in order to define the variety $\Vbf * \Dbf$.  We
finally present a bibliography giving, for each fragment $\Fs$ in
\figurename~\ref{fig:frag}, references for solving the above
questions~\ref{item:13}--\ref{item:15}. The statement and the proof of the transfer result,
Theorem~\ref{thm:main}, is postponed to Section~\ref{sec:algebra}.

\subsection{Varieties}
\label{sec:varieties}

A \emph{variety of semigroups (resp. monoids)} is a class of finite semigroups
(resp. monoids) closed under three natural operations: finite direct product,
subsemigroup (or submonoid), and homomorphic image. This makes it possible to
define classes of regular languages based on the monoids that recognize these
languages: a variety \Vbf defines the class of all languages
recognized by semigroups (resp. monoids) in~\Vbf. There is an issue however:
all classes of languages defined in this way have to be closed under
complement, since the set of languages recognized by any semigroup is closed
under complement. This prevents us from capturing logical fragments that are
not closed under complement, such as \sdo. This problem has been solved
in~\cite{porder} with the notions of \emph{ordered semigroups and
  monoids}. Intuitively, such a semigroup is parametrized by a partial order
and the set of languages it recognizes is then restricted with respect to this
partial order.

Let us recall this notion, which leads to the definition of variety of ordered
semigroups or monoids. All classes considered in this paper may be defined in
terms of such varieties.

\highlight{Ordered Semigroups} An ordered semigroup is a pair $(S,\leqslant)$
where $S$ is a semigroup and~$\leqslant$ is a partial order on $S$, which
is compatible with multiplication: $s \leqslant t$ and $s' \leqslant t'$ imply
$ss' \leqslant tt'$. To simplify the notation, we will often omit the
partial order $\leqslant$ when it is clear from the context and simply
speak of an ordered semigroup $S$. Observe that any semigroup endowed
with equality as the partial order is an ordered semigroup. In
particular we view $A^+$ as an ordered semigroup with equality as the
partial order.

If $(S,\leqslant_S)$ and $(T,\leqslant_T)$ are ordered semigroups, an ordered
semigroup morphism is a mapping $\alpha: S \rightarrow T$ which
is a semigroup morphism and preserves the partial order, \emph{i.e.}, for all
$s,s' \in S$, $s \leqslant_S s' \Rightarrow \alpha(s) \leqslant_T \alpha(s')$. Let
$L \subseteq A^+$ and $(S,\leqslant)$ be an ordered semigroup. Then, $L$ is
said to be \emph{recognized by $(S,\leqslant)$} if there exist an ordered
semigroup morphism $\alpha: A^+ \rightarrow S$ and $F \subseteq S$,
such that $L = \alpha^{-1}(F)$ and $F$ is \emph{upward closed}, that is:
\[
s \in F \text{ and } s \leqslant t \Rightarrow t \in F.
\]
When $\leqslant$ is trivial, then any subset of $S$ is upward closed, and we
recover exactly the classical notion of recognizability by semigroups
presented just above. However, when $\leqslant$ is nontrivial, the set of
recognized languages gets restricted because of the additional condition on
the recognizing set~$F$. In particular it may happen that a language is
recognized by $(S,\leqslant)$, while its complement is not (its complement is
recognized by $(S,\geqslant)$).

\highlight{Varieties of Ordered Semigroups} A \emph{variety of finite
  ordered semigroups} is a class \Vbf of finite ordered semigroups
that satisfies the following properties:
\begin{enumerate}
\item \Vbf is closed under ordered subsemigroup: if $(S,\leqslant) \in
  \Vbf$, then $(T,\leqslant) \in \Vbf$ when $T$ is a subsemigroup of $S$ and the order
  on $T$ is the restriction of the order on $S$.
\item \Vbf is closed under ordered quotient: if $(S,\leqslant) \in
  \Vbf$ and $\alpha: (S,\leqslant) \rightarrow (T,\leqslant)$ is a surjective ordered semigroup
  morphism, then we have $(T,\leqslant) \in \Vbf$.
\item \Vbf is closed under Cartesian direct product: if $(S_1,\leqslant_1),(S_2,\leqslant_2) \in \Vbf$,
  then we have $(S_1\times S_2,\leqslant)\in\Vbf$, where the semigroup $S_1 \times S_2$ is
  equipped with the componentwise multiplication and $(s_1,s_2) \leqslant
  (t_1,t_2)$ if $s_1\leqslant_1t_1$ and $s_2\leqslant_2t_2$.
\end{enumerate}

Note that for technical reasons, we have to consider both varieties of
semigroups and monoids: non-enriched fragments correspond to varieties of
monoids while enriched ones correspond to varieties of semigroups. For the
sake of simplifying the presentation, we only give the definitions for
semigroups. Ordered monoids and varieties of ordered monoids are defined in a
similar way, as well as the non-ordered versions.

% \begin{remark}
%   While it is standard to define varieties with these three conditions,
%   the only property that is truly needed for the main theorem to hold is
%   for \Vbf to be closed under cartesian products.
% \end{remark}

\highlight{Varieties and Classes of Languages} To any variety \Vbf of ordered
semigroups (resp.~of ordered monoids), we can associate the class of all languages
that are recognized by an ordered semigroup (resp.~ordered monoid) in
\Vbf. As for logics and for the sake of simplifying the presentation, we may
abuse notation and use \Vbf to denote both a variety and the class of
languages it defines.

It turns out that all classes from \figurename~\ref{fig:frag} can be defined
in such a way. Therefore, they all have an associated a variety. This follows
actually from a general result, Eilenberg's theorem.  One should however keep
in mind that in this framework, there is:

\begin{enumerate}[label=$(\alph*)$]
\item\label{item:16} Eilenberg's theorem, a generic result establishing a
  correspondence between varieties and classes of languages (indexed by
  alphabets) enjoying certain closure properties: closure under Boolean
  operations, inverse morphisms and left and right residuals. It
  was first obtained by S.~Eilenberg for classes closed under complement, and
  later generalized by J.E.~Pin~\cite{porder} when this assumption does not
  necessarily hold.

\item\label{item:17} Specific instances of Eilenberg's theorem, one for each
  particular class, relating such a class of languages
  with a corresponding variety of ordered semigroups or monoids.
\end{enumerate}

We will not state Eilenberg's theorem precisely, as we do not need it. On the
other hand, Item~\ref{item:17} is useful to provide an alternate version of
our transfer result, Theorem~\ref{thm:sdo}, in the algebraic framework of
Section~\ref{sec:algebra}. This alternate version, Theorem~\ref{thm:main}, is
generic, in the sense that it transfers decidability of the separation problem
for a variety $\Vbf$ to the variety $\Vbf * \Dbf$, with no assumption on the
variety \Vbf. However, in order to instantiate this generic theorem for our
logical fragments, we need Item~\ref{item:15} above, \emph{i.e.}, to show that
for each weak fragment~\Fs, if the variety associated to \Fs is \Vbf, then the
variety associated to the enriched variant~$\Fs^+$ is $\Vbf\ast\Dbf$. In other
words, we shall rely on the aforementioned specific connections,
Item~\ref{item:17} above, between a class of languages and a variety of
ordered semigroups or monoids. Each fragment will be described in
Section~\ref{sec:algebraic-charac}, and the fact that for all of them, if \Fs
corresponds to the variety \Vbf, then $\Fs^+$ corresponds to the variety
$\Vbf * \Dbf$ is stated in Theorem~\ref{thm:vstard}.

\subsection{The Semidirect Product}
\label{sec:semidirect-product}

Let $M$ be an ordered monoid and let $T$ be an ordered semigroup. A
\emph{semidirect product} of $M$ and $T$ is an operation which is parametrized
by an \emph{action} of $T$ on $M$ and outputs a new ordered semigroup, whose
base set is $M\times T$. In particular, one can obtain different semidirect
products out of the same $M$ and $T$, depending on the chosen~action.

Let `$+$' and `$\cdot$' be the operations of $M$ and $T$ respectively. Note
that we choose to denote the operation on $M$ additively. This is for the sake
of simplifying the presentation. However, this does not mean that we assume
$M$ to be commutative. An \emph{action} `$\ast$' of $T$ on $M$ is a mapping
$(t,s) \mapsto t \ast s$ from $T^1 \times M$ to $M$ such that, for all
$s,s' \in M$ and all $t,t' \in T$: %
\begin{multicols}{2}\ignorespaces
  \begin{itemize}[itemsep=.8mm,topsep=0mm,parsep=0mm,partopsep=0mm]
  \item $t \ast (t' \ast s) = (t \cdot t') \ast s$.
  \item $1_T \ast s = s$.
  \item if $s \leqslant s'$, then $t \ast s \leqslant t \ast s'$.
  \item $t \ast (s + s') = t \ast s + t \ast s'$.
  \item $t \ast 1_M = 1_M$.
  \item if $t \leqslant t'$, then $t \ast s \leqslant t' \ast s$.
  \end{itemize}
\end{multicols}
Given a fixed action `$\ast$' of $T$ on $M$, the \emph{semidirect
  product  $M \ast T$ of $M$ and $T$ with respect to action $\ast$}\/ is the set
$M \times T$ equipped with the following operation:
\[
(s,t) \cdot (s',t') = (s + t \ast s',t \cdot t')
\]
\noindent
and the componentwise order:
\[
(s,t) \leqslant (s',t') \text{ if } s \leqslant s' \text{ and }
t \leqslant t'.
\]
One can verify that this does yield an ordered
semigroup, see~\cite{semidirect-ordered:2002}.

Given a variety \Vbf of ordered monoids and a variety \Wbf of ordered
semigroups, we denote by $\Vbf\ast\Wbf$ the variety of ordered semigroups
generated by all semidirect products of the form $M\ast T$, with $M\in\Vbf$
and $T\in\Wbf$, where $\ast$ ranges over all possible actions of $T$ on~$M$.

\highlight{The Variety \Dbf}
We will only use the semidirect product with semigroups $T$ from a specific
variety, denoted by $\Dbf$. This is because such a semidirect product
$\Vbf\ast\Dbf$ of $\Vbf$ with $\Dbf$ is often related to the enrichment with
the successor relation of the fragment captured by $\Vbf$.

The variety \Dbf consists of all finite ordered semigroups $S$ such that for
all $s \in S$ and all $e \in E(S)$, we have $se = e$. From a language
perspective, a language $L$ is recognized by a semigroup in \Dbf iff there
exists $k \in \nat$ such that membership of a word $w$ to~$L$ only depends on
the suffix of length $k$ of $w$.

The reason why we introduce such semidirect products is the following theorem,
which gathers several nontrivial results from the literature listed in
Section~\ref{sec:algebraic-charac}, and which answer our
requirement~\ref{item:15} towards our transfer theorem.

\begin{theorem}
  \label{thm:vstard}
  Let \Vbf be a variety corresponding to a fragment \Fs from the ones
  presented in \figurename~\ref{fig:frag}. Then, the variety corresponding to the
  fragment $\Fs^+$ is $\Vbf*\Dbf$.
\end{theorem}

\subsection{Algebraic Characterizations of Logically Defined Fragments}
\label{sec:algebraic-charac}

In this section, we consider Items~\ref{item:14} and \ref{item:15}, which were
to be solved in order to apply our generic theorem. All logical fragments of
\figurename~\ref{fig:frag} correspond to varieties that have been fully
identified.  We present, for each such fragment, bibliographic references
relating its weak and strong variants to varieties. In particular, we will see
that Theorem~\ref{thm:vstard} holds: for each fragment whose non-enriched
variant corresponds to a variety \Vbf of ordered monoids, its enriched version
corresponds to the variety of ordered semigroups $\Vbf * \Dbf$ built
from~$\Vbf$.

\subsubsection{First-order with Equality}
\label{sec:foeq}
The logic \foeq is the restriction of \fow in which the linear order cannot be
used, and only equality between two positions can be tested. It is folklore
that \foeq-definable languages are exactly those that can be defined using a
monoid in the variety of monoids {\sf ACom} of aperiodic and commutative
monoids.

The enriched fragment is $\foeqp$, as $min$ and $max$ can be
eliminated in the formulas. It defines locally threshold testable
languages~\cite{Thom82}. In~\cite{TherienWeiss:ltt-wreath:1985}, it was proved
that \foeqp-definable languages are exactly those that can be defined in
${\sf ACom} * \Dbf$. In particular this was used to solve the membership
problem for \foeqp.

That separation is decidable for \foeq is simple (essentially, the
problem can be reduced to the decision of Presburger logic,
see~\cite{ltltt:2013}). Hence Theorem~\ref{thm:sdo} and
Theorem~\ref{thm:main} yield two different proofs of the following
corollary.

\begin{corollary} \label{cor:ltt}
  Let $L,L'$ be regular languages. It is decidable to test whether $L$
  is \foeqp-separable from $L'$.
\end{corollary}

As we already explained, while the proof of Corollary~\ref{cor:ltt} is new,
the result itself is not. A specific proof was presented in~\cite{ltltt:2013}
and the result can also be obtained through indirect means by combining
results from~\cite{MR1709911,Steinberg:delay-pointlikes:2001}.

\subsubsection{Quantifier Alternation Hierarchy}
\label{sec:quant-altern-hier}
One can classify first-order formulas by counting the number of alternations
between $\exists$ and $\forall$ quantifiers in the prenex normal form of the
formula. For $i \in \nat$, a formula is said to be \sio{i} (resp.\ \pio{i}) if
its prenex normal form has $(i -1)$ quantifier alternations (that is, $i$
blocks of quantifiers) and starts with an $\exists$ (resp.~ a\ $\forall$)
quantifier. For example, a formula whose prenex normal form is
\[
\exists x_1 \exists x_2  \forall x_3 \exists x_4
\ \varphi(x_1,x_2,x_3,x_4) \quad \text{(with $\varphi$ quantifier-free)}
\]
\noindent
is {\sio 3}. Observe that a \pio{i} formula is by definition the
negation of a \sio{i} formula. Finally, a \bso{i} formula is a boolean
combination of \sio{i} formulas.

Both this hierarchy and the enriched variant are known to be
strict~\cite{BroKnaStrict,ThomStrict}. Furthermore, they correspond to
well-known hierarchies of classes of languages: the non-enriched hierarchy
corresponds to the Straubing-Thérien hierarchy~\cite{StrauConcat,TheConcat},
while the enriched hierarchy corresponds to the dot-depth
hierarchy~\cite{BrzoDot}. Note that for all fragments above \sdo, the
predicates $min$ and $max$ can be eliminated from the logic. Hence, we denote
the enriched fragments by {\sip 1}, \mbox{{\bsp 1}}, $\mbox{\sdp},\dots$

Solving the membership problem for all levels in both hierarchies has been an
open problem for a long time. As of today, only the lower levels are known to
be decidable. Historically, {\bso 1} and {\bsp 1} have been investigated
first. It is known from~\cite{simon75} that {\bso 1} has decidable membership
and corresponds to the variety of monoids~{\sf J}. For {\bsp 1}, decidability
was proved in~\cite{Knast:dd1:1983a}, as well as the correspondence with the
variety of semigroups ${\sf J} * \Dbf$ in~\cite{Str85}.

The fragments {\sio 1} and {\sio 2} were shown to have decidable
membership in~\cite{pwdelta}. Moreover, the authors also prove that
each of these two fragments correspond to varieties of ordered monoids and that
{\sip 1} and \sdp correspond to the varieties of semigroups obtained
by taking the semidirect product with \Dbf. From this correspondence,
they obtain decidability of {\sip 1}. This is more involved for \sdp
and was proved later in~\cite{glasser-dd3/2}.

Recently, membership has been shown to be decidable for both {\bso 2}
and {\sio 3}~\cite{PZ:icalp14}. These results can be transferred to
\bdp and \stp using a result by Straubing~\cite{Str85}, or
Theorem~\ref{thm:memb} in this paper. For all levels
above, the membership problem is open.

Separation is known to be decidable for {\sio 1}~\cite{sep_icalp13},
{\bso   1}~\cite{DBLP:conf/mfcs/PlaceRZ13,sep_icalp13} and
\sdo~\cite{PZ:icalp14}. Hence Theorem~\ref{thm:sdo} and
Theorem~\ref{thm:main} yield two different proofs of the following
corollary.

\begin{corollary}
  Let $L,L'$ be regular languages, then the following problems are
  decidable:
  \begin{itemize}
  \item whether $L$ is $\sip 1$-separable from $L'$.
  \item whether $L$ is $\bsp 1$-separable from $L'$.
  \item whether $L$ is \sdp-separable from $L'$.
  \end{itemize}
\end{corollary}

As we explained in Section~\ref{sec:prelims}, the result for {\bsp 1}
as it can also be obtained through indirect means by
combining results
from~\cite{MR1709911,Steinberg:delay-pointlikes:2001}. On the other
hand, the results are new for both $\sip 1$ and \sdp.

\subsubsection{Two-Variable First-Order Logic}
\label{sec:two-variable-first}
The logic \fod is the restriction of \fow using only two (reusable)
variables. The corresponding enriched fragment is \fodp ($min$ and $max$ can
be eliminated from the logic).

In~\cite{TW-FO2}, it was proved that \fod and \fodp correspond
respectively to the varieties {\sf DA} and ${\sf DA} * \Dbf$. This
immediately yields decidability of membership for \fod. For \fodp,
this additionally requires a deep algebraic result by
Almeida~\cite{Almeida:1996c} (a simpler self-contained proof also
exists~\cite{PSDAD}). The separation problem has been proved to be
decidable for \fod in~\cite{DBLP:conf/mfcs/PlaceRZ13}. Hence
Theorem~\ref{thm:sdo} and Theorem~\ref{thm:main} yield two different
proofs of the following corollary.

\begin{corollary} \label{cor:dodp}
  Let $L,L'$ be regular languages. It is decidable to test whether $L$
  is \fodp-separable from $L'$.
\end{corollary}

As we explained in Section~\ref{sec:prelims}, while the proof is new, the
result itself is not. It can also be obtained through indirect means, again by
combining results from~\cite{MR1709911,Steinberg:delay-pointlikes:2001}.

\section{Algebraic Approach}
\label{sec:algebra}

\medskip We are now ready to prove Theorem~\ref{thm:main}. Recall that we have
a non-trivial variety \Vbf of ordered monoids, two languages $L$ and $L'$
recognized by a morphism $\alpha: A^+ \rightarrow S$, and $\Lbb,\Lbb'
\subseteq \wfA^+$ the associated languages of well-formed words.

We prove that $L$ is $(\Vbf * \Dbf)$-separable from $L'$ if and only if
\Lbb is \Vbf-separable from $\Lbb'$. We prove each direction in its
own subsection.

We now present an algebraic version of Theorem~\ref{thm:sdo}: the operator
$\mathsf{V} \mapsto \mathsf{V}\ast\mathsf{D}$ preserves decidability of
separation.

We would like to emphasize again that the ideas behind this theorem are
essentially the same as for Theorem~\ref{thm:sdo}. In particular, proofs only
rely on elementary notions, thus bypassing complex constructions usually used
to prove this kind of result, even if the statement itself requires some
additional algebraic~vocabulary.

\smallskip
The section is organized in three parts.
\begin{itemize}
\item We first briefly recall how classes of languages corresponding to our
  logical fragments are given an algebraic definition: for each fragment, an
  associated class of finite semigroups (or monoids)~\Vbf, a \emph{variety},
  has already been characterized, such that the class of languages definable
  in the~fragment is exactly the class of languages that are recognized by a
  semigroup (or monoid) of \Vbf.
\item In the second part, we define what ``adding
  the successor relation'' means in this context. Given a variety \Vbf, this
  generally corresponds to considering a new variety built on top of~\Vbf via
  an operation called the \emph{semidirect product}. This new variety is
  denoted $\Vbf * \Dbf$.
\item  Finally, in the last part, we state our main
  theorem: for any variety \Vbf, separability for the variety $\Vbf * \Dbf$
  reduces to separability for the variety~\Vbf.
\end{itemize}

\subsection{Main Theorem}

We have now the machinery needed to state our main theorem. For any variety of
ordered monoids \Vbf, we reduce $(\Vbf * \Dbf)$-separability to
$\Vbf$-separability.

\begin{theorem}
 \label{thm:main}
  Let \Vbf be a non-trivial variety of ordered monoids. Let $L$ and $L'$
  be two languages both recognized by the same morphism $\alpha: A^+ \rightarrow S$
  into a finite semigroup~$S$. Set $\Lbb,\Lbb' \subseteq \wfA^+$ as the
  languages of well-formed words associated to $L,L'$, respectively. Then, $L$ is
  $(\Vbf * \Dbf)$-separable from $L'$ if and only if\/ $\Lbb$ is
  \Vbf-separable from $\Lbb'$.
\end{theorem}

In view of Theorem~\ref{thm:vstard}, Theorem~\ref{thm:main} applies to
all fragments we introduced. This means that Theorem~\ref{thm:sdo} can
be given an alternate indirect proof within this algebraic framework by
combining Theorem~\ref{thm:main} and Theorem~\ref{thm:vstard}. Hence,
this also yields another proof of Corollary~\ref{cor:sdp}.

The proof of Theorem~\ref{thm:main} is presented in the rest of this
section. As it was the case for Theorem~\ref{thm:sdo}, the proof is both
elementary and constructive: if there exists a separator for \Lbb and $\Lbb'$
in \Vbf, we use it to construct a separator for $L$ and $L'$ in $\Vbf * \Dbf$.

This rest of the section is divided in three parts. In the first one, we
recall the formal definition of the semidirect product operation. In the next
two ones, we prove both directions of Theorem~\ref{thm:main}.

\subsection{\texorpdfstring{From $(\Vbf * \Dbf)$-separability to
    \Vbf-separability}{From (V * D)-separability to V-separability}}

We prove that if $L$ is $(\Vbf*\Dbf)$-separable from $L'$, then
$\Lbb$ is \Vbf-separable from $\Lbb'$. Note that we reuse the
construction which associates a canonical word $\ucroch{\wbb}_i \in
A^+$ to every word $\wbb \in \wfA^+$ and natural $i \geq 1$ (see
Section~\ref{sec:from-fodp-to-fod} for details).

Assume that $L$ is $(\Vbf*\Dbf)$-separable from $L'$. This means that there
exists an element of $(\Vbf*\Dbf)$ separating $L$ and $L'$. By
\cite[Prop.~3.5]{semidirect-ordered:2002}, such an ordered semigroup is an
ordered quotient of an ordered subsemigroup of a semidirect product $M * T$, with $M
\in \Vbf$ and $T \in \Dbf$. Therefore, $M * T$ itself separates $L$ and $L'$.
Hence, there is some upward closed $F \subseteq M * T$ and a morphism
$\delta: A^+ \rightarrow M * T$ such that $\delta^{-1}(F)$ separates
$L$ from $L'$.

We construct a separator in \Vbf for \Lbb and $\Lbb'$. Set $T =
\{t_1,\dots,t_n\}$ and observe that since $\Vbf$ is non-trivial, it
contains an ordered monoid $N$ containing at least $n$ distinct
elements. We choose $n$ such elements $t'_1,\dots,t'_n$ of $N$. The
choice is essentially arbitrary, but we ask $t'_1,\dots,t'_n$ to be
pairwise incomparable with respect to the partial order $\leqslant$. We
prove that \Lbb can be separated from $\Lbb'$ using the ordered monoid
$\Mbb = M \times N \in \Vbf$ (recall that a variety is closed under
Cartesian product). For an element $t=t_i$ of $T$, we denote by $t'$
the element $t'_i$ of $N$.

We define a morphism $\gamma: \wfA^+ \rightarrow \Mbb$ as follows. Let
$\omega$ be the idempotent power
$\omega(M * T)$ of $M * T$. Set $\abb = (e,s,f) \in \wfA$, so that
$\ucroch{\abb}_\omega=w_e^\omega w_s^{}w_f^\omega$. Let
$\delta(w_e^\omega)=(m_e,t_e)\in M*T$ and
$\delta(\ucroch{\abb}_\omega)=(m,t)$. We
define $\gamma(\abb)\in M\times N$ as follows:
\[
  \gamma(\abb) = \left\{\begin{array}{ll} (t_e\ast m,\ t') & \text{when $f = 1_S$,} \\
  (t_e\ast m,\ 1_N) & \text{otherwise.}\end{array}\right.
\]
This defines a morphism $\gamma:\wfA\to M\times N\in \Vbf$. It remains to
prove that $\gamma$ recognizes a separator of $\Lbb$ and $\Lbb'$. This is a
consequence of the next lemma.

\begin{lemma} \label{lem:msowillneverdie}
Let $\wbb \in \wfA^+$ be well-formed, and set $(m,t_i) =
\delta(\ucroch{\wbb}_\omega)$. Then $\gamma(\wbb) = (m,t'_i)$.
\end{lemma}

Before proving the lemma, we use it to conclude the proof. Define
$\Fbb \subseteq \Mbb$ by $\Fbb=\{(m,t'_i) \mid (m,t_i) \in F\}\}$.
One can verify that \Fbb is upward closed. We claim that
$\gamma^{-1}(\Fbb)$ separates \Lbb from $\Lbb'$.

Assume first that $\wbb \in \Lbb$. By Fact~\ref{fct:cons1},
$\ucroch{\wbb}_\omega \in L$, hence $\delta(\ucroch{\wbb}_\omega) \in
F$. It then follows from Lemma~\ref{lem:msowillneverdie} that
$\gamma(\wbb) \in \Fbb$. Conversely if $\wbb \in \Lbb'$, we have
$\delta(\ucroch{\wbb}_\omega) \not\in F$. It then follows from
Lemma~\ref{lem:msowillneverdie} that $\gamma(\wbb) \not\in \Fbb$ which
terminates the proof. We now prove Lemma~\ref{lem:msowillneverdie}.

\begin{proof}[Proof of Lemma~\ref{lem:msowillneverdie}.]
  We first show that the first component in $M$ of
  $\delta(\ucroch{\wbb}_\omega)$ and of $\gamma(\wbb)$ are equal. The proof
  consists in a straightforward but tedious computation. Set $\wbb =
  \abb_1\cdots\abb_p \in \wfA^+$ that is well-formed. Set
  $\abb_i=(e_{i-1},s_i,e_i)$ and recall that, in view of the definition of
  $\ucroch{\wbb}_i$ given in Section~\ref{sec:from-fodp-to-fod}, we have
  chosen words $w_{e_i}$ and $w_{s_i}$ such that:
  \[
  \ucroch{\abb_i}_\omega=w_{e_{i-1}}^\omega\cdot w_{s_i}\cdot w_{e_{i}}^\omega
  \]
  For each idempotent $e=e_i$, set $\delta(w_e^\omega)=(m_e,t_e)\in M*T$
  and for each element $s=s_i$, let $\delta(w_s)=(m_s,t_s)\in M*T$. Note that
  by definition of $\omega$, the element $(m_e,t_e)=\delta(w_e^\omega) =
  \delta(w_e)^\omega$ is idempotent, so $(m_e,t_e)=(m_e+t_e\ast
  m_e,t_e^2)$. In particular, $t_e$ is idempotent in $T$. Further, we have for all $i$:
  \begin{equation}
    \label{eq:idem}
    m_{e_i}+t_{e_i}\ast m_{e_i} = m_{e_i}.
  \end{equation}
  % To lighten the notation, from now on, let us write the action of $T$ on $M$ as $tm$ instead
  % of $t\ast m$.
  For each $\abb_i=(e_{i-1},s_i,e_i)$, we then have
  \begin{align}
    \delta(\ucroch{\abb_i}_\omega) &= \delta(w_{e_{i-1}}^\omega)
    \delta(w_{s_i})\delta(w_{e_i})^\omega\notag{}\\
    &= (m_{e_{i-1}},t_{e_{i-1}})(m_{s_i},t_{s_i})(m_{e_i},t_{e_i})\notag{}\\
    &= \Bigl(m_{e_{i-1}}+t_{e_{i-1}}\ast m_{s_i}+t_{e_{i-1}}t_{s_i}\ast m_{e_i},\quad t_{e_i}\Bigr)\label{eq:2}
  \end{align}
  where, for computing the 2nd component, we used the fact that $t_{e_i}$
  is idempotent in $T\in \Dbf$.
  Similarly, by definition we have $\ucroch{\wbb}_\omega=(w_{e_0})^\omega
  w_{s_1}(w_{e_1})^\omega \cdots(w_{e_{p-1}})^\omega
  w_{s_p}(w_{e_p})^\omega$, and
  \begin{align}
    \label{eq:3}
    \delta(\ucroch{\wbb}_\omega) &= \delta(w_{e_0}^\omega)
    \delta(w_{s_1})\delta(w_{e_1})^\omega w\cdots \delta(w_{e_{p-1}})^\omega
    \delta(w_{s_p})w_{s_p}(w_{e_p})^\omega\notag{}\\
    &= (m_{e_0},t_{e_0})(m_{s_1},t_{s_1})(m_{e_1},t_{e_1})\cdots
    (m_{e_{p-1}},t_{e_{p-1}})(m_{s_p},t_{s_p})(m_{e_p},t_{e_p})\notag{}\\
    &=
    \Bigl(m_{e_0}+(t_{e_0}\ast m_{s_1}+t_{e_0}t_{s_1}\ast m_{e_1})+\cdots+(t_{e_{p-1}}\ast m_{s_p}+t_{e_{p-1}}t_{s_p}\ast m_{e_p}),\quad
    t_{e_p}\Bigr).
  \end{align}
  Again, for the last equality, we used the definition of the semidirect
  product and the fact that each $t_{e_i}$ is an idempotent in $T$, which
  implies, since $T\in\Dbf$, that $t\cdot t_{e_i}=t_{e_i}$ for all $t\in T$.

  Using \eqref{eq:idem}  for each $i$, one can replace $m_{e_i}$
  in~\eqref{eq:3} by
  $m_{e_i}+t_{e_i}\ast m_{e_i}$. Taking into account that $t_{e_i}$ is idempotent
  in $T$, this yields for this first component of
  $\delta(\ucroch{\wbb}_\omega)$ the value
$$m_{e_0}+t_{e_0}\ast m_{e_0}+(t_{e_0}\ast m_{s_1}+t_{e_0}t_{s_1}\ast m_{e_1}+t_{e_1}\ast m_{e_1})+\cdots+(t_{e_{p-1}}\ast m_{s_p}+t_{e_{p-1}}t_{s_p}\ast m_{e_p}+t_{e_p}\ast m_{e_p})$$

  Observe that since $\wbb$ is well-formed, $e_0=1_S$, hence $m_{e_0}=1_M$,
  which is the neutral element for the `$+$' operation on $M$. In the same
  way, $e_p=1_S$, hence $m_{e_p}=1_M$, and therefore, using the last axiom of
  an action, we deduce that $t_{e_p}\ast m_{e_p}=1_M$. Hence, these two elements
  can be removed from the expression of the first component of
  $\delta(\ucroch{\wbb}_\omega)$. Therefore, this first component can be rewritten,
  using associativity, as:
  \begin{equation}
    \label{eq:4}
    (t_{e_0}\ast m_{e_0}+t_{e_0}\ast m_{s_1}+t_{e_0}t_{s_1}\ast m_{e_1})
    +\cdots+
    (t_{e_{p-1}}\ast m_{e_{p-1}}+t_{e_{p-1}}\ast m_{s_p}+t_{e_{p-1}}t_{s_p}\ast m_{e_p}).
  \end{equation}

  On the other hand, in view of~\eqref{eq:2} and by definition of $\gamma$,
  the first component of $\gamma(\abb_i)\in M\times N$ is
  \begin{align}
    (t_{e_{i-1}}\ast (m_{e_{i-1}}+t_{e_{i-1}}\ast m_{s_i}+t_{e_{i-1}}t_{s_i}\ast m_{e_i})
    &= (t_{e_{i-1}}\ast m_{e_{i-1}}+t_{e_{i-1}}\ast m_{s_i}+t_{e_{i-1}}t_{s_i}\ast m_{e_i}).\label{eq:5}
  \end{align}
  Therefore, one can compute the first component of
  $\gamma(\wbb)=\gamma(\abb_1\cdots\abb_p)=\gamma(\abb_1)\cdots\gamma(\abb_p)$
  by summing the values~\eqref{eq:5} for $i=1,\ldots,p$ (recall that the
  operation on $M$ is noted additively), which gives the value computed in \eqref{eq:4}.
  Hence we have shown that the first component in $M$ of
  $\delta(\ucroch{\wbb}_\omega)$ and of $\gamma(\wbb)$ are equal.

  \medskip It remains to check that when the second component of
  $\delta(\ucroch{\wbb}_\omega)$ is equal to some $t\in T$, then the second
  component of $\gamma(\wbb)$ is the corresponding element $t'\in N$. This is
  simpler: by definition of a well-formed word, we have $e_i\neq1_S$ for
  $i<p$, and $e_p=1_S$. By definition of $\gamma$, it follows that the second
  component of $\gamma(\wbb)$ is the second component of $\gamma(\abb_p)$,
  namely $t'_p$. Now, since $T\in\Dbf$, the second component of
  $\delta(\ucroch{\abb}_\omega)$ is $t_p$, which concludes the proof.
\end{proof}

\subsection{\texorpdfstring{From \Vbf-separability to
    $(\Vbf*\Dbf)$-separability}{From V-separability to (V * D)-separability}}

We prove that if $\Lbb$ is \Vbf-separable from $\Lbb'$, then $L$ is
$(\Vbf*\Dbf)$-separable from $L'$. Note that we reuse the construction which
associates to every word $w \in A^+$ a canonical word $\croch{w} \in \wfA^+$
(see Section~\ref{sec:from-fod-to-fodp} for details).

Assume that \Lbb is \Vbf-separable from $\Lbb'$. This means that we
have a morphism $\gamma: \wfA^* \rightarrow \Mbb$ with \Mbb an ordered
monoid in \Vbf and $\Fbb \subseteq \Mbb$ upward-closed such that $\gamma^{-1}(\Fbb)$
separates \Lbb from $\Lbb'$. We need to construct a separator in $\Vbf
* \Dbf$ for $L$ and $L'$. The main idea is to define a morphism, which
given $w \in A^+$, computes $\gamma(\croch{w})$. This is slightly
technical however as the morphism needs some machinery to make this
computation.

We begin with some notations. To every word $w \in A^+$, we associate an
element $lab(w) \in \Mbb$. Let $x$ be the last position in $w$
and consider the construction of $\croch{w}$. If $x$ is distinguished,
we set $lab(w) = \gamma(\abb)$ with $\abb$ the label of $\croch{x}$ in
$\croch{w}$. Otherwise, we simply set $lab(w) = 1_{\Mbb}$. We can now
start the construction of our separator. We have to define the
following objects:
\begin{itemize}
\item An ordered semigroup $T \in \Dbf$.
\item An ordered monoid $M \in \Vbf$.
\item An action of $T$ on $M$ yielding a semidirect product $M
  * T$.
\item A morphism $\delta : A^+ \rightarrow M * T$ which recognizes
  the desired separator.
\end{itemize}

\highlight{Definition of $T$} We set $T$ as the set $\{w \in A^+ \mid |w|
\leqslant 2|S|\}$ equipped with the following operation. If $w,w'\in T$,
we set $w \cdot w'$ as the suffix of length $2|S|$ of the word $ww'$
when $ww'$ has length $\geq 2|S|$ and as $ww'$ otherwise. One can
verify that this operation is indeed associative and that $T \in
\Dbf$. We use equality as the partial order on $T$.

Observe that we have a natural morphism $\rho: A^+ \rightarrow T$
such that $\rho(w)$ is $w$ if $|w| \leqslant 2|S|$, and $\rho(w)$ is the suffix of
length $2|S|$ of $w$ otherwise. Observe that by Lemma~\ref{lem:canonic},
we have the following fact.

\begin{fact} \label{fct:canonic}
For every $w \in A^+$, $lab(w) = lab(\rho(w))$.
\end{fact}

\highlight{Definition of $M$} We set $M \in \Vbf$ as the Cartesian
product $\Mbb^{T^1}$ (recall that as a variety of ordered monoids,
\Vbf is closed under Cartesian product).

\begin{remark}
Since we intend to take a semidirect product of $M$ and $T$, we will
denote the semigroup operations of both $M$ and $\Mbb$ additively in
order to clarify the presentation.
\end{remark}

\highlight{Definition of $M * T$} If $w \in T$ and $f \in M$ (\emph{i.e.}, $f$
is a mapping $f: T^1 \rightarrow \Mbb$), we set $w \cdot f$ as the
mapping $g: T^1 \rightarrow M$ such that $g(u) = f(u \cdot w)$. One
can verify that '$\cdot$' is an action of $T$ on $M$. In the remainder
of the proof, we denote by $M * T$ the semidirect product of $M$ and
$T$ with respect to this action.

\highlight{Definition of $\delta$} Set $f_{Id}: T^1 \rightarrow \Mbb$
defined as follows. We set $f_{Id}(1_T) = 1_\Mbb$ and $f_{Id}(w) =
lab(w)$ when $w \in T$. We can now define $\delta: A^+ \rightarrow M *
T$. Let $a \in A^+$, we set $\delta(a)$ as the pair $(f_a,a)$ where
$f_a = a \cdot f_{Id}$, \emph{i.e.}, the mapping $f_a: w \mapsto
f_{Id}(wa)$. It now remains to prove that $\delta$ does recognize a
separator of $L$ from $L'$. This is a consequence of the following
lemma.

\begin{lemma} \label{lem:thisistheend}
Let $w \in A^+$, $(f,u) = \delta(w)$ and $end(u)$ as the label of the
last position in $\croch{u}$. Then,
\[
\gamma(\croch{w}) = f(1_T) \cdot \gamma(end(u)).
\]
\end{lemma}
We first use the lemma to conclude the proof. Set $F \subseteq M *
T$ as the set \[F = \{(f,u) \mid f(1_T) \cdot \gamma(end(u)) \in
\Fbb\}.\] One can verify that $F$ is upward closed (this is essentially
because \Fbb is upward-closed). It is immediate from
Lemma~\ref{lem:thisistheend} that $\delta(w) \in F$ iff
$\gamma(\croch{w}) \in \Fbb$. We claim that $\delta^{-1}(F)$ separates
$L$ from $L'$.

Assume first that $w \in L$, we need to prove that $w \in
\delta^{-1}(F)$. By Fact~\ref{fct:sametype}, we have $\croch{w} \in
\Lbb$, hence $\gamma(\croch{w}) \in \Fbb$ and $\delta(w) \in
F$. Similarly, if $w \in L'$, $\croch{w} \in \Lbb'$, hence
$\gamma(\croch{w})  \not\in \Fbb$ and $\delta(w) \not\in F$ which
terminates the proof. It finally remains to prove
Lemma~\ref{lem:thisistheend}.

\begin{proof}[Proof of Lemma~\ref{lem:thisistheend}.]
Set $w = a_1 \cdots a_n$ and $\croch{w} = \abb_1\cdots \abb_m$. By
definition, we have:
\[
f = \rho(a_1) \cdot f_{Id} + \rho(a_1a_2) \cdot f_{Id} + \cdots +
\rho(a_1\cdots a_n) \cdot f_{Id}
\]
By definition of $f_{Id}$ and by Fact~\ref{fct:canonic} this means
that:
\[
f(1_T) = lab(a_1) + lab(a_1a_2) + \cdots + lab(a_1\cdots a_n)
\]
It is then immediate from the definition of \croch{w} that $f(1_T) =
\gamma(\abb_1 \cdots \abb_{m-1})$. Hence $\gamma(\croch{w}) = f(1_T)
\cdot \gamma(end(w))$. This finishes the proof since $u$ is the suffix
of length $2|S|$ of $w$, and therefore $end(u) = end(w)$ by
Lemma~\ref{lem:canonic}.
\end{proof}

% \section{Infinite words}
% \label{sec:infinite-words}

\section{Conclusion}
\label{sec:conc}

We proved that separation is decidable over finite words for the following
logical fragments: \foeqp, \mbox{{\sip 1}}, {\bsp 1}, \sdp and \fodp. To
achieve this, we presented a simple reduction to the same problem for the
weaker fragments \foeq, \mbox{{\sio 1}}, {\bso 1}, \sdo and \fod.

The reduction itself is entirely generic to all fragments and its proof is
elementary, and also mostly generic. In particular, the technique can be used
to prove that the reduction works for other natural fragments of first-order
logic. An interesting example to which these results apply is the quantifier
alternation hierarchy within \fod (known as the Trotter-Weil hierarchy, and
which is decidable~\cite{kufleitner-weil:fo2-lmcs}). However, the separation
problem for classes in this hierarchy has yet to be investigated. We also
obtained direct proofs that membership is decidable for $\bdp$ and $\stp$.

Finally, we presented an algebraic formulation of this reduction, which recovers
a previously known result by Steinberg~\cite{Steinberg:delay-pointlikes:2001},
while having a much simpler proof. One can expect extending these results to
other fragments, such as enrichment with modulo predicates.
Another advantage of this technique is that it can be extended in
a straightforward way to the same logical fragments over words of
infinite length. This yields identical transfer results. We leave
the presentation of these results for further work.

% \bibliographystyle{plain}
% \bibliography{po-full}

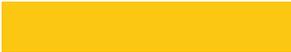
\end{document}